\documentclass[aps,pre,twocolumn,showkeys,superscriptaddress,floatfix
]{revtex4-2}

\usepackage{amsmath,amssymb,amsthm,amscd,latexsym}

\usepackage{graphicx}
\usepackage{color}
\usepackage{tikz}
\usetikzlibrary{arrows}
\usepackage{subfigure}
\usepackage{epsfig}
\graphicspath{{./}}

\usepackage{array}
\usepackage[section]{placeins}
\usepackage{enumitem}
\usepackage{float}
\usepackage{algorithm}
\usepackage{algpseudocode}
\usepackage{listings}
\usepackage{verbatim}

\usepackage[utf8]{inputenc}
\usepackage{indentfirst}
\usepackage[colorlinks=true,allcolors=blue]{hyperref}

\newtheorem{proposition}{Proposition}[section]


\begin{document}

\title{Higher-Order Co-Infection on Hypergraphs}

\author{Katrina Liu}
\affiliation{
Canyon Crest Academy, San Diego, California, USA
}

\author{Laura P.\ Schaposnik}
\affiliation{
MSCS, LQuTE, AIRSHIP Lab, 
University of Illinois Chicago, USA
}
\affiliation{
NSF-Simons National Institute for Theory and Mathematics in
Biology, Chicago, IL, USA
}
\affiliation{
Mathematical Institute and Magdalen College,
University of Oxford, UK
}

\begin{abstract}
People encounter infections through overlapping groups, such as
households, classrooms, and workplaces. Understanding how these
shared contacts affect the spread of multiple pathogens is critical to determining what information epidemic models need to retain.
We introduce HyCo-SIS (Higher-order Co-infection SIS), a
discrete-time SIS model on hypergraphs with explicit reinforcement
between overlapping active groups and pairwise transmission.
We investigate how this reinforcement affects infection dynamics
and whether simplified representations of hypergraphs preserve the
resulting behavior. We derive an exact representation of overlap
reinforcement as target-dependent four-person interactions and
determine when a symmetric representation is possible. Numerical
experiments examine reinforcement strength, contact structure,
and susceptibility to secondary infection, as well as calibrated
pairwise approximations. Our results demonstrate that redistributing
overlap weights can substantially change finite-time co-infection
even when every individual's total overlap coefficient is preserved.
We further derive an exact five-person calculation and demonstrate
why equal mean infection pressure does not necessarily produce
equal acquisition probabilities. These findings identify how the
allocation of higher-order interactions can influence contagion
dynamics beyond their aggregate strength. Beyond epidemiology,
this work contributes to the statistical physics of spreading
processes by examining how the arrangement of many-body interactions
shapes collective dynamics, and what information is lost when those
interactions are simplified.
\end{abstract}


\maketitle
\raggedbottom

\setcounter{tocdepth}{2}
 \twocolumngrid
\raggedbottom
\setlength{\textfloatsep}{10pt plus 2pt minus 2pt}
\setlength{\dbltextfloatsep}{10pt plus 2pt minus 2pt}
\setlength{\floatsep}{8pt plus 2pt minus 2pt}
\setlength{\dblfloatsep}{8pt plus 2pt minus 2pt}
\section{Introduction}

Infections rarely travel alone, and they rarely travel along single contacts. A child can bring one virus home from a classroom and another from a sports team, and the same classmate may sit in several of those groups at once. This picture has inspired much research over the last few decades (see, e.g., \cite{pastorsatorras2015epidemicreview,keeling2005networks,battiston2020higherorder,iacopini2019simplicial,sahneh2014competitive,poletto2015competing,stonge2021school} and references therein), and it has run through much of the second author's earlier work: from modelling epidemics on $d$-cliqued graphs, where people meet in tight overlapping clusters rather than one at a time \cite{zhang2018dcliqued}, to percolation for mutating infectious diseases \cite{luo2020percolating} and age-structured models for COVID-19 type viruses \cite{ram2021agesir}, to correlations between COVID-19 and dengue \cite{SchaposnikEtAl2023SciRep} and, more recently, network models of super-infection, co-infection and hotspot growth for two pathogens \cite{SchaposnikEtAl2025a,SchaposnikEtAl2025b}.

Mathematical models let us ask how this structure shapes transmission, but what they predict depends on how the structure is represented. Counting a person's contacts does not tell us who meets together, or which groups share members. The question behind this paper is therefore a simple one: which features of group contact must a model keep before it can reliably compare spreading outcomes?

Consider again students who meet in overlapping classroom, lunch and study groups. The same infectious student can appear in several of another student's group exposures, and since group transmission depends on the collective state of the group, rearranging who meets with whom can change infection probabilities even when every individual's contact total stays the same. Figure \ref{fig:mechanism} shows a small example with five students and three overlapping groups.




\begin{figure}[!t]
\centering
\begin{tikzpicture}[scale=1.15,
  every node/.style={font=\small},
  contact/.style={draw=black!55,line width=0.9pt}]
\definecolor{pathA}{HTML}{2A6F97}
\definecolor{pathB}{HTML}{C0623F}
\definecolor{fillLunch}{HTML}{F6E6D8}
\definecolor{fillStudy}{HTML}{D4E8E4}
\definecolor{fillClub}{HTML}{E6E7F2}
\coordinate (n0) at (2.9,3.6);
\coordinate (n1) at (1.4,0.9);
\coordinate (n2) at (3.7,1.1);
\coordinate (n3) at (0.5,3.9);
\coordinate (n4) at (5.2,3.9);
\fill[fillLunch] (n3)--(n0)--(n1)--cycle;
\fill[fillStudy] (n0)--(n1)--(n2)--cycle;
\fill[fillClub]  (n0)--(n2)--(n4)--cycle;
\draw[contact] (n3)--(n0) (n3)--(n1) (n0)--(n1) (n1)--(n2) (n0)--(n2) (n0)--(n4) (n2)--(n4);
\node[fill=white,fill opacity=.7,text opacity=1,inner sep=1pt] at (1.6,2.6) {$g_2$: lunch};
\node[fill=white,fill opacity=.7,text opacity=1,inner sep=1pt] at (2.65,1.95) {$g_1$: study};
\node[fill=white,fill opacity=.7,text opacity=1,inner sep=1pt] at (3.95,2.95) {$g_3$: club};
\def\r{3.2mm}
\fill[pathA] (n0) circle (\r);  \fill[white] (n1) circle (\r);
\fill[pathA] (n3) circle (\r);  \fill[pathB] (n4) circle (\r);
\fill[pathB] (n2) circle (\r);
\fill[pathA] (n2) ++(0,\r) arc[start angle=90,end angle=270,radius=\r]--cycle;
\foreach \n in {n0,n1,n2,n3,n4} \draw[black!80,line width=0.9pt] (\n) circle (\r);
\node[above=3.5mm] at (n0) {$0$};
\node[below=3.5mm] at (n1) {$1$ (target)};
\node[below=3.5mm] at (n2) {$2$};
\node[above=3.5mm] at (n3) {$3$};
\node[above=3.5mm] at (n4) {$4$};
\begin{scope}[shift={(-0.35,-0.55)},every node/.style={font=\footnotesize,anchor=west}]
\def\s{1.3mm}
\fill[white] (0,0) circle (\s); \draw[line width=0.6pt] (0,0) circle (\s); \node at (0.12,0) {susceptible};
\fill[pathA] (1.75,0) circle (\s); \draw[line width=0.6pt] (1.75,0) circle (\s); \node at (1.87,0) {pathogen $1$};
\fill[pathB] (3.65,0) circle (\s); \draw[line width=0.6pt] (3.65,0) circle (\s); \node at (3.77,0) {pathogen $2$};
\fill[pathB] (5.55,0) circle (\s);
\fill[pathA] (5.55,\s) arc[start angle=90,end angle=270,radius=\s]--cycle;
\draw[line width=0.6pt] (5.55,0) circle (\s); \node at (5.67,0) {both};
\end{scope}
\end{tikzpicture}
\caption{Five students in three overlapping groups,
$g_1:=\{0,1,2\}$, $g_2:=\{0,1,3\}$, and
$g_3:=\{0,2,4\}$. For susceptible student~1, students
0, 2, and 3 carry pathogen~1. The study and lunch groups
are therefore both active for this target and share
infectious source~0. Student~1 receives three pairwise
contributions, two active-group contributions, and one
overlap contribution. At the experimental pairwise,
triangle, and reinforcement rates of $0.005$, $0.1$,
and $0.4$, respectively, the total infection pressure is
$3(0.005)+2(0.1)+0.4=0.615$, giving acquisition
probability $1-e^{-0.615}\approx0.459$.
}
\label{fig:mechanism}
\end{figure}
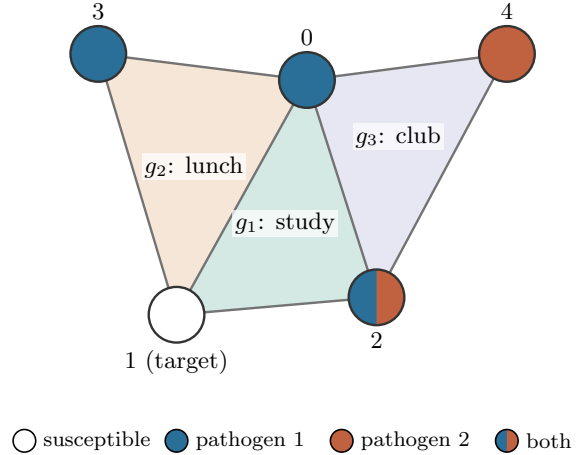

In this paper we introduce a two-pathogen discrete-time model based on the standard SIS (Susceptible-Infected-Susceptible) model with collective activation and additional reinforcement from overlapping active groups, \emph{HyCo-SIS} (Higher-order Co-infection SIS). The model focuses on three transmission channels: exposure through pairwise contacts, collective exposure within triangle groups, and reinforcement between overlapping active triangles. For a given pathogen, a triangle is active for a target when both other members carry that pathogen. The overlap term $\gamma$ adds infection pressure when two such active triangles share another member besides the target, while a separate parameter $\alpha$ scales susceptibility to acquiring a second pathogen. 

The model also poses a question in statistical physics about the relationship between microscopic interactions and collective behavior. For each pathogen, carriage and recovery correspond to activation and deactivation of a site. Activity can spread through pairwise contacts and collective group activation, and the same total interaction strength can be distributed among different sets of neighbors. We ask whether redistribution changes their collective dynamics. Further, by comparing the model with simplified versions of it, we examine the limits of coarse-graining and ask which details can be discarded while preserving the expected outcome.

Classical network epidemic models provide a foundation for this investigation. In standard pairwise SIS and SIR formulations, individuals are represented by nodes and transmission opportunities by edges \cite{pastorsatorras2015epidemicreview,newman2002spread,keeling2005networks}. Hypergraphs extend this description by allowing an interaction to involve more than two individuals; simplicial complexes additionally include the lower-dimensional faces of each simplex \cite{battiston2020higherorder,benson2018simplicial,torres2021simplicial}, supporting infection rules that depend on several sources acting together. Past research demonstrated that collective transmission over simplicial contagion models can produce discontinuous transitions and bistability in settings where its pairwise counterpart does not \cite{iacopini2019simplicial}. Subsequent work has examined heterogeneous hypergraphs, higher-order spreading with indirect pathways, and the effects of hyperedge overlap \cite{landry2020effect,cui2024general,malizia2025overlap}.

Two further aspects of research motivate our comparisons. Studies of interacting epidemics show how competition, cooperation, and altered susceptibility can affect simultaneous spreading processes \cite{sahneh2014competitive,poletto2015competing,granell2013dynamical}. In addition, research on cluster synchronization shows that distinct hypergraphs can share the same pairwise projection, so the projected graph need not retain enough information to determine the original higher-order dynamics \cite{salova2021projection}. Nevertheless, suitably adjusted pairwise epidemic models can approximate higher-order aggregate trajectories \cite{tan2026replicating}. Exact representation of microscopic transitions and approximate prediction of selected observables are therefore different questions. Our model HyCo-SIS builds on these established ideas, and we investigate a specified overlap rule and how its allocation among source groups affects spreading.

\paragraph*{Summary of results.}
Our main findings can be summarized as follows.
\begin{itemize}
\item {\bf Model and exact structure.} We introduce HyCo-SIS, a two-pathogen SIS model on hypergraphs with pairwise, triangle and overlap-reinforcement channels, and show that overlap reinforcement is exactly a sum of target-dependent four-person interactions, characterizing when it admits a symmetric representation (Sec. \ref{subsec:overlap-model}, \ref{subsec:four-person}).
\item {\bf Limits of pairwise information.} In the single-pathogen
mean-hazard closure, the higher-order channels do not enter the
disease-free linearization, so local stability is governed by the
pairwise channel. The reported ensemble-mean spectral ratios lie
below this closure's pairwise stability boundary
(Sec.~\ref{subsec:linearization},\ref{subsec:seed}), and we show that weighted pairwise
projection need not determine the higher-order transition
probabilities (Sec.~\ref{subsec:nonidentifiability}).
\item {\bf Seeds, overlap and projections.} Outbreaks depend   on introduction size, overlap reinforcement raises co-infection when triangles share edges. The fitted pairwise projections do not reproduce terminal-window co-infection and outbreak frequency on the held-out graphs.(Sec.~\ref{subsec:seed}, \ref{subsec:overlap-results},\ref{subsec:projections}).
\item {\bf Allocation matters.} Redistributing overlap weights changes finite-time co-infection
substantially even when each individual's total overlap
coefficient is preserved. An exact five-person calculation
provides a one-step mechanism consistent with the observed
acquisition differences, through the nonlinear conversion of
infection pressure into acquisition probability (Sec.~\ref{subsec:allocation-results},\ref{subsec:five-person}).
\end{itemize}

The remainder of this paper is organized as follows. Section \ref{sec:background} reviews related work on higher-order contagion and graph projections. Section \ref{sec:model} defines HyCo-SIS, its transmission and recovery rules, and the pairwise representations used for comparison. Section \ref{sec:analytic} presents the disease-free linearization under a mean-hazard closure and an exact example of information lost through pairwise projection. Section \ref{sec:numerics} reports the numerical experiments on network structure, initial infection levels, reinforcement, and pairwise approximations. Section \ref{sec:allocation} develops the four-person representation of overlap reinforcement, examines two weight redistribution controls, and presents the identical state diagnostics and exact five-person calculation. Finally, Section \ref{sec:discussion} summarizes the findings, discusses possible extensions to empirical contact data, and places the model in the context of statistical physics.

Our goal is to determine which features of higher-order exposure are lost when interaction weights are redistributed or projected onto pairwise contacts. We compare the original overlap rule with a symmetric control preserving group totals and a second control preserving target totals, at $\alpha=1$ and $2$. Calibrated pairwise projections provide a complementary comparison.

\section{Background and Related Work}
\label{sec:background}

A hypergraph is a mathematical structure used to record group memberships, marking which individuals belong to the same group. A related structure is a simplicial complex, which requires downward closure, where every subset of a simplex must also be included. Previous work has studied collective activation rules, in which infection is triggered only when several members of a group are infected together; generalized nonlinear diffusion, in which spreading departs from simple linear laws; and the effect of structural heterogeneity on epidemic outcomes \cite{berge1984hypergraphs,iacopini2019simplicial,kiss2024exact,cui2024general,landry2020effect,stonge2021school}. Building on these studies, we compare group-level information at the microscopic scale with the predictions obtained after aggregating it into pairwise contacts.

Another parallel line of research focuses on interacting infectious diseases and the co-infection process, where individuals are infected with multiple pathogens simultaneously. Two-pathogen models capture competition for susceptible hosts, cross-immunity, in which infection with one pathogen protects against the other, and mutual reinforcement, in which each pathogen helps the other spread \cite{sahneh2014competitive,poletto2015competing,granell2013dynamical}. Related studies explore general hypergraph diffusion models adaptable to various scenarios \cite{cui2024general}, as well as multi-layer reaction-diffusion models specifically designed for multi-pathogen scenarios \cite{SchaposnikEtAl2025a,SchaposnikEtAl2025b}. Here we analyze the specific overlap rule of HyCo-SIS.

Previous literature has established that projecting a hypergraph onto an ordinary pairwise graph can result in substantial information loss. Structurally distinct hypergraphs may yield identical adjacency matrices after projection while supporting markedly different system dynamics \cite{salova2021projection}. Structural features specific to higher-order systems, including overlap among groups, can likewise affect contagion outcomes in ways that may not be retained under pairwise representations \cite{malizia2025overlap}. At the same time, appropriately calibrated pairwise epidemic models can approximately reproduce the aggregate trajectories of higher-order models \cite{tan2026replicating}. Thus, the local counterexample presented in this study is not intended as a novel demonstration that projection can lose information, but rather to isolate the specific mechanism by which information is lost under the transmission rules considered here. The central numerical question is instead whether a finely calibrated pairwise projection can generalize beyond the observable used for calibration and accurately reproduce other outcomes of the higher-order system.

Related work has also considered directed contagion processes \cite{directedcontagion2024}, while symmetrization procedures designed to preserve group coupling strength have long been used as controls in studies of network synchronization \cite{directedsync2022}. Against this broader background, our analysis focuses specifically on the coefficient values induced by the active-triangle pair rule and on how projection redistributes these coefficients throughout the system. The resulting structural constraints are closely related to established concepts in hypergraph combinatorics \cite{beaudou2013lines}. Together, these observations motivate our comparison between the original higher-order model and calibrated pairwise projections: rather than asking whether projection can reproduce a single aggregate trajectory, we ask which higher-order observables remain identifiable after calibration and which are irretrievably altered by the projection.

We address three related questions. First, how do reinforcement
strength and the allocation of overlap coefficients among source
groups affect finite-time co-infection on static hypergraphs?
Second, does a two-pathogen pairwise model calibrated to endpoint
union prevalence reproduce co-infection, outbreak frequency,
and temporal dynamics on held-out graphs? Third, does a weighted
pairwise projection determine the higher-order transition
probabilities from a fixed labeled carrier configuration?

\section{HyCo-SIS: Model set up}
\label{sec:model}
This section formally defines the hypergraph contagion model used in our numerical experiments. We first introduce the single-pathogen SIS-type baseline model on hypergraphs, then extend it to a two-pathogen co-infection model that includes synergy or competition parameters. Finally, we define three pairwise projections of the hypergraph that serve as baselines for comparison.

\subsection{Hypergraph framework}

Let $
H=(V,E)
$
be a finite hypergraph, where
\[
V:=\{1,2,\ldots,n\}
\]
is the set of nodes and \(E\subseteq 2^V\setminus\{\emptyset\}\) is the set of hyperedges. A hyperedge \(e\in E\) represents a group interaction among all nodes in \(e\). For example, a hyperedge can correspond to all members of a student group, where students who join the same group are enclosed by the same hyperedge. In this study, a hyperedge of size 2 corresponds to a one-on-one contact, that is, an ordinary interaction between two people; a hyperedge of size 3 or larger corresponds to a higher-order interaction, that is, a group contact involving more than two people.

For any node $i\in V$, denote $E_i:=\{e\in E: i\in e\}$ as the set of all hyperedges that contain node $i$. Time evolves in discrete steps $t\in\{0,1,2,\ldots\}$.
Let $X_i(t)\in\{0,1\}$ indicate whether node $i$ carries
the pathogen at step $t$, with $1$ denoting carriage
and $0$ denoting susceptibility. Write
$X(t):=(X_1(t),\ldots,X_n(t))$.
The number of carriers in hyperedge $e$, excluding
target $i$, is defined by
\[
I_e^{(-i)}(t):=\sum_{j\in e\setminus\{i\}}X_j(t),
\] the value of which represents the number of already infected nodes among all remaining members of the hyperedge $e$ that contains node $i$, excluding node $i$ itself.

\subsection{Single-pathogen hypergraph contagion}

Let us first consider the single-pathogen baseline without overlap reinforcement ($\gamma=0$), following the classical SIS framework on networks \cite{pastorsatorras2015epidemicreview,barrat2008dynamical} and its extensions to higher-order interaction structures including simplicial complexes and hypergraphs \cite{iacopini2019simplicial,cui2024general}.

Using the carrier indicators introduced above, an infected
node recovers independently with probability $\mu$ at
each step.  
Thus, $\mu\in[0,1]$ is interpreted as the per-step recovery probability.

Susceptible nodes receive infection pressure through pairwise contacts
and higher-order group exposure. In the numerical experiments, $A$ is
the binary adjacency matrix of all distinct pairwise contacts, including
every triangle side and every additional pairwise contact. Thus
$A_{ij}=1$ when distinct nodes $i$ and $j$ share a pairwise contact,
and $A_{ij}=0$ otherwise, with $A_{ii}=0$. Repeated membership of a
pair in different triangles does not increase $A_{ij}$.  The pairwise infection pressure on node $i$ is
\[
\beta_1\sum_{j=1}^n A_{ij}X_j(t),
\]
where $\beta_1$ controls the strength of pairwise transmission.

To characterize higher-order exposure, we assume that a hyperedge contributes infection pressure only when a sufficient number of other nodes in that hyperedge are infected. For an integer threshold $r\geq2$, define the activation
indicator below, where $\mathbf{1}(\cdot)$ equals one
when its condition holds and zero otherwise:
\[
\phi_{e}^{(r)}(i,t)
:=
\mathbf{1}\!\left(I_e^{(-i)}(t)\geq r\right).
\]
Then the higher-order infection pressure on node $i$ is
\[
\beta_h
\sum_{\substack{e\in E_i\\ |e|\geq 3}}
\phi_{e}^{(r)}(i,t),
\]
where $\beta_h$ controls the strength of higher-order transmission.

Combining these two parts together, the total infection pressure that a susceptible node $i$ experiences at time $t$ is
\[
\lambda_i(t)
:=
\beta_1\sum_{j=1}^n A_{ij}X_j(t)
+
\beta_h
\sum_{\substack{e\in E_i\\ |e|\geq 3}}
\mathbf{1}\!\left(I_e^{(-i)}(t)\geq r\right).
\]

In all triangle-based simulations, we set $r=2$, meaning that only when both other nodes in the triple are infected does the triple generate higher-order infection pressure for the remaining node.

The quantity $\lambda_i(t)$ represents the total infection pressure experienced by node $i$ during time step $t$. The probability that a susceptible node is infected in this time step is

\[
p_i(t)
:=
1-\exp\!\left[-\lambda_i(t)\right].
\]
This maps any nonnegative pressure to a probability in $[0,1)$.

We distinguish throughout between the infection pressure
\(\lambda_i(t)\), which is nonnegative and need not be bounded by one, and
the corresponding per-step infection probability \(p_i(t)\in[0,1]\).
All infection pressures at time $t$ are computed from the configuration of the system at the beginning of that time step.

\subsection{Two-pathogen co-infection model}
\label{subsec:two-pathogen}

We next extend the single-pathogen process to a two-pathogen discrete-time
co-infection model, following the general framework of interacting and
co-circulating spreading processes
\cite{sahneh2014competitive,poletto2015competing,granell2013dynamical,SchaposnikEtAl2023SciRep}.

Any node \(i\in V\) occupies one of four states, $S, I_1, I_2, I_{12},$ where \(S\) denotes susceptibility to both pathogens, \(I_1\) denotes
infection with pathogen 1 only, \(I_2\) denotes infection with pathogen 2
only, and \(I_{12}\) means the node is infected with both pathogens.

For pathogen \(a\in\{1,2\}\), let $\lambda_i^{(a)}(t)$ denote the infection pressure, or hazard, experienced by node \(i\) from
pathogen \(a\). The corresponding infection probability is
\[
p_i^{(a)}(t)
:=
1-\exp\!\left[-\lambda_i^{(a)}(t)\right].
\]

If a node is already infected with one of the two pathogens, the probability that it becomes infected with the remaining pathogen (the secondary infection probability) is adjusted by an interaction parameter $\alpha>0$.  Thus, for a node in
state \(I_1\), the probability of acquiring pathogen 2 during the time step is
\[
\widetilde p_i^{(2)}(t)
:=
1-\exp\!\left[-\alpha\lambda_i^{(2)}(t)\right],
\]
Conversely, for a node in state \(I_2\), the probability of acquiring pathogen 1 is
\[
\widetilde p_i^{(1)}(t)
:=
1-\exp\!\left[-\alpha\lambda_i^{(1)}(t)\right].
\]

Recovery from pathogen \(a\) occurs independently with probability
\(\mu_a\) per time step.

All infection pressures and transition probabilities at time \(t\) are
computed from the configuration at the beginning of the time step.
Infection and recovery events are then sampled, and all node states are
updated simultaneously to produce the configuration at time \(t+1\).
In particular, if a susceptible node acquires both pathogens in the same
time step, it enters state \(I_{12}\).

Conditional on the configuration at time \(t\), the infection-acquisition
and recovery events appearing in Table \ref{tab:transition-rules} are
sampled independently. Infection-acquisition events are sampled only for
pathogens absent from a node at time \(t\), while recovery events are
sampled only for pathogens present at time \(t\). As a result, if a node clears a certain pathogen during this round of synchronous updates, it cannot become re-infected with that same pathogen in the same round of updates, and must wait until the next time step to have the opportunity to be infected again.
The complete synchronous transition rule is given in
Table \ref{tab:transition-rules}.

\begin{table}[!t]
\centering
\small
\begin{tabular}{c|l|c}
\hline
State at \(t\) & Events during the step & At \(t+1\) \\
\hline

\(S\)
& neither infection occurs
& \(S\) \\

\(S\)
& pathogen 1 only
& \(I_1\) \\

\(S\)
& pathogen 2 only
& \(I_2\) \\

\(S\)
& both infections occur
& \(I_{12}\) \\
\hline

\(I_1\)
& no pathogen-2 infection, no recovery
& \(I_1\) \\

\(I_1\)
& no pathogen-2 infection, recovery from 1
& \(S\) \\

\(I_1\)
& pathogen-2 infection, no recovery from 1
& \(I_{12}\) \\

\(I_1\)
& pathogen-2 infection and recovery from 1
& \(I_2\) \\
\hline

\(I_2\)
& no pathogen-1 infection, no recovery
& \(I_2\) \\

\(I_2\)
& no pathogen-1 infection, recovery from 2
& \(S\) \\

\(I_2\)
& pathogen-1 infection, no recovery from 2
& \(I_{12}\) \\

\(I_2\)
& pathogen-1 infection and recovery from 2
& \(I_1\) \\
\hline

\(I_{12}\)
& neither pathogen clears
& \(I_{12}\) \\

\(I_{12}\)
& pathogen 1 clears only
& \(I_2\) \\

\(I_{12}\)
& pathogen 2 clears only
& \(I_1\) \\

\(I_{12}\)
& both pathogens clear
& \(S\) \\
\hline
\end{tabular}
\caption{Synchronous transition rules for the two-pathogen HyCo-SIS model.
All probabilities are evaluated from the configuration at the beginning of
time step \(t\).}
\label{tab:transition-rules}
\end{table}

\subsection{Synergy and competition}

The interaction between the two pathogens is controlled by a fixed parameter $\alpha>0$.

The parameter \(\alpha\) alters only the scenario of secondary infection (infection
by one pathogen of a node already carrying the other pathogen), as defined in Section \ref{subsec:two-pathogen}. The three interaction regimes are therefore
\[
0<\alpha<1
\qquad\text{competition or partial cross-protection},
\]
\[
\alpha=1
\qquad\text{no cross-susceptibility modification},
\]
and
\[
\alpha>1
\qquad\text{synergy or co-facilitation}.
\]

Hence, increasing \(\alpha\) raises the probability of a node developing a secondary infection
without directly altering the initial infection pressure faced by a
node susceptible to both pathogens.

When studying the coexistence of two pathogens, the main observable is the proportion of hosts infected with both pathogens, 
\[
\rho_{12}(t)
:=
\frac{1}{n}
\#\{i\in V:i\text{ is in state }I_{12}\text{ at time }t\}.
\]
Let $I_1(t)$, $I_2(t)$, and $I_{12}(t)$ denote the
numbers of nodes in the corresponding states at step $t$.
We define the total infected fraction by
\[
\rho_{\mathrm{union}}(t)
:=\frac{I_1(t)+I_2(t)+I_{12}(t)}{n},
\]
so that each infected node is counted once.
The experiments of Section \ref{sec:numerics} use $\alpha=1$, the case in which the pathogens do not interact, and the allocation experiments of Section \ref{sec:allocation} also use $\alpha=2$. In both we ask how higher-order structure and the overlap between hyperedges affect $\rho_{12}(t)$.

\subsection{Overlap-reinforced higher-order co-infection}
\label{subsec:overlap-model}
For pathogen \(a\in\{1,2\}\), define the indicator
\[
X_j^{(a)}(t)
:=
\begin{cases}
1, & \text{if \(j\) carries \(a\) at step \(t\)},\\
0, & \text{otherwise}.
\end{cases}
\]

Let $N_i^{(a)}(t)
:=
\sum_{j=1}^{n}A_{ij}X_j^{(a)}(t)$ denote the number of pairwise neighbors of node \(i\) carrying pathogen \(a\).

For a hyperedge \(e\in E_i\), define
\[
I_{e,a}^{(-i)}(t)
:=
\sum_{j\in e\setminus\{i\}}
X_j^{(a)}(t).
\]
For a threshold \(r\geq2\), define the activation indicator
\[
\phi_{e,a}^{(r)}(i,t)
:=
\mathbf{1}
\left(
I_{e,a}^{(-i)}(t)\geq r
\right).
\]

The number of infectious higher-order hyperedges containing node \(i\) is
\[
R_i^{(a)}(t)
:=
\sum_{\substack{e\in E_i\\|e|\geq3}}
\phi_{e,a}^{(r)}(i,t).
\]

To quantify overlap among infectious higher-order hyperedges, fix any ordering of $E$, so that $e<f$ below counts each unordered pair once, and define
\[
O_i^{(a)}(t)
:=
\sum_{\substack{e,f\in E_i\\e<f\\|e|,|f|\geq3}}
\left|
(e\setminus\{i\})
\cap
(f\setminus\{i\})
\right|
\phi_{e,a}^{(r)}(i,t)\phi_{f,a}^{(r)}(i,t).
\]

To calculate the overlap statistic, we remove node \(i\) from all intersection sets to prevent the trivial common
membership of \(i\) from contributing to the overlap statistic.  Thus
\(O_i^{(a)}(t)\) only counts overlap among the other members of simultaneously
infectious groups containing \(i\).

The total infection pressure by pathogen \(a\) on node \(i\) is
\[
\begin{aligned}\lambda_i^{(a)}(t):={}&\beta_a^{(1)}N_i^{(a)}(t)+\beta_a^{(h)}R_i^{(a)}(t)\\ &+\gamma_aO_i^{(a)}(t).\end{aligned}
\]

The three coefficients in this formula regulate the transmission intensity in different scenarios respectively: \(\beta_a^{(1)}\) controls pairwise transmission,
\(\beta_a^{(h)}\) controls higher-order transmission, and
\(\gamma_a\geq0\) controls reinforcement due to overlapping infectious
hyperedges.

The three channels contribute additively to the infection pressure.
Their exposure statistics may be correlated because they are evaluated
on the same carrier configuration. For example, if
\(e=\{i,j,k\}\) is a triangle and both \(j\) and \(k\) carry pathogen \(a\),
then node \(i\) receives the two corresponding pairwise contributions through
\(N_i^{(a)}\) as well as the additional higher-order contribution through
\(R_i^{(a)}\). This additive structure is deliberate: pairwise contact and
simultaneous group exposure are distinct transmission mechanisms.

Substituting the infection pressure above into the primary and secondary
infection probabilities defined in Section \ref{subsec:two-pathogen} gives
\[
p_i^{(a)}(t)
=
1-\exp\!\left[-\lambda_i^{(a)}(t)\right],
\]
\[
\widetilde p_i^{(a)}(t)
=
1-\exp\!\left[-\alpha\lambda_i^{(a)}(t)\right].
\]

We refer to this process as
\emph{HyCo-SIS}, the overlap-reinforced higher-order co-infection SIS model.

Since all our experiments use the same parameters for both pathogens, we write
\[
\beta_1^{(1)}=\beta_2^{(1)}=\beta_1,\qquad
\beta_1^{(h)}=\beta_2^{(h)}=\beta_h,
\]
\[
\gamma_1=\gamma_2=\gamma,
\qquad
\mu_1=\mu_2=\mu.
\]
The two pathogens therefore differ only in their initial carriers.

\subsection{Graph projections used for comparison}
The projections are built from the binary contact matrix
$A$ defined above and the triangle collection
\[
E_\Delta:=\{e\in E:|e|=3\}.
\]
For distinct vertices $i,j$, define
\[
M_{ij}:=\#\{e\in E_\Delta:\{i,j\}\subseteq e\},
\qquad M_{ii}:=0.
\]
Thus $M_{ij}$ counts the distinct triangles containing
both vertices. The evaluated projections are
\[
C:=A,\qquad W:=A+M,\]
\[
Q_{ij}:=\mathbf 1\{W_{ij}\geq2\},\quad i\ne j,
\]
with all diagonal entries set to $0$, so that no node interacts with itself. Thus $W$ counts both the size-two contact and the triangle memberships. A contact outside every triangle has weight one. Each projected model uses only pairwise infection pressure $\widehat\beta\sum_j P_{ij}X_j^{(a)}$, where $P$ is its matrix.  

All projected models use the same recovery rule and the same synchronous update as HyCo-SIS; they differ from it only through the matrix $P$ and the coefficient $\widehat\beta$, which is fitted separately for each projection (Section \ref{subsec:projections}).

\begin{algorithm}[H]
\caption{HyCo-SIS: Overlap-Reinforced Higher-Order Co-Infection SIS}
\begin{algorithmic}[1]

\State \textbf{Input:}
hypergraph \(H=(V,E)\);
parameters
\(\beta_a^{(1)},\beta_a^{(h)},\gamma_a,\mu_a\) for
\(a\in\{1,2\}\);
interaction parameter \(\alpha\);
activation threshold \(r\);
initial node states;
time horizon \(T\)

\State Initialize every node in one of
\(S,I_1,I_2,I_{12}\)

\For{\(t=0,1,\ldots,T-1\)}

    \State Store the complete configuration at time \(t\)

    \For{each node \(i\in V\)}

        \For{each pathogen \(a\in\{1,2\}\)}

            \State Compute the pairwise exposure
            \(N_i^{(a)}(t)\)

            \State Compute the number of activated higher-order
            hyperedges \(R_i^{(a)}(t)\)

            \State Compute the infectious overlap
            \(O_i^{(a)}(t)\)

            \State Compute $\lambda_i^{(a)}(t)$ from $N_i^{(a)},R_i^{(a)},O_i^{(a)}$.

            \State Compute the primary infection probability
\[
p_i^{(a)}(t)
=
1-\exp\!\left[-\lambda_i^{(a)}(t)\right]
\]

\State Compute the secondary infection probability
\[
\widetilde p_i^{(a)}(t)
=
1-\exp\!\left[-\alpha\lambda_i^{(a)}(t)\right]
\]
            \State For each pathogen absent from node \(i\) at time \(t\), use
\(p_i^{(a)}(t)\) if node \(i\) is in state \(S\), and use
\(\widetilde p_i^{(a)}(t)\) if node \(i\) carries the other pathogen

\State For pathogen \(a\) present in node \(i\) at time \(t\), set
the recovery probability to \(\mu_a\)

        \EndFor

        \State Sample each applicable infection-acquisition and recovery
        event exactly once, independently conditional on the configuration
        stored at time \(t\)
        
        \State Determine the new state of node \(i\) according to
        Table \ref{tab:transition-rules}

    \EndFor

    \State Apply all node-state changes simultaneously

    \State Record
    \(S(t+1),I_1(t+1),I_2(t+1),I_{12}(t+1)\)

\EndFor

\State \textbf{Output:}
state history and epidemic observables

\end{algorithmic}
\end{algorithm}

\section{Analytical Results}
\label{sec:analytic}
In this section we collect two analytical observations about HyCo-SIS. The first concerns the disease-free state: under a mean-hazard closure we linearize the dynamics and show that only the pairwise channel survives at first order, so the higher-order mechanisms can only act once enough carriers are present. The second concerns graph projections: we give an explicit pair of hypergraphs with the same weighted pairwise projection but different transition probabilities, which isolates the information that a projection discards.

\subsection{Disease-free linearization under a mean-hazard closure}
\label{subsec:linearization}
Let $x_i(t)$ denote the infection probability of node $i$
at step $t$ in the single-pathogen restriction, and write
$x(t):=(x_1(t),\ldots,x_n(t))$. We explicitly use a mean-hazard closure to obtain an update formula for change in infection probability,
\[
\begin{aligned}x_i(t+1)={}&(1-\mu)x_i(t)\\&+(1-x_i(t))\{1-e^{-h_i(x(t))}\}.\end{aligned}
\]

The first term accounts for infected nodes that do not recover; the
second accounts for acquisition by susceptible nodes. We write
\[
h_i(x):=\beta_1\sum_j A_{ij}x_j+H_i(x).
\]
To construct $H_i$, we first simplify products of binary carrier
indicators using $X_j^2=X_j$, and then replace each product over
distinct vertices by the product of their infection probabilities.
For distinct triangle hyperedges with $r=2$, an active triangle
contributes a quadratic term. A pair of triangles sharing the target
and one source contributes a cubic term, because
$(X_jX_k)(X_jX_l)=X_jX_kX_l$. Thus
$H_i(x)=O(\|x\|^2)$ near the disease-free state,
where $\|\cdot\|$ denotes the Euclidean norm.
We write $I$ for the identity matrix and $\rho(A)$
for the spectral radius of $A$.

\begin{proposition}\label{prop:dfe}
Assume $\beta_1,\beta_h,\gamma\geq0$, $0\leq\mu\leq1$, and that
$A$ is symmetric and nonnegative. For the mean-hazard closure above,
the Jacobian at the disease-free fixed point is
\[
J=(1-\mu)I+\beta_1A,
\]
with spectral radius
\[
\rho(J)=1-\mu+\beta_1\rho(A).
\]
The disease-free fixed point is locally asymptotically stable if
$\beta_1\rho(A)<\mu$ and unstable if $\beta_1\rho(A)>\mu$.
The case $\beta_1\rho(A)=\mu$ is not decided by the linearization.
\end{proposition}

\begin{proof}
Since $h_i(0)=0$ and $DH_i(0)=0$, differentiation of the closure
update at zero gives $J=(1-\mu)I+\beta_1A$.
The matrix $J$ is nonnegative. If $v$ is a nonnegative eigenvector
of $A$ associated with $\rho(A)$, then
$Jv=(1-\mu+\beta_1\rho(A))v$. Moreover, every eigenvalue
$\lambda$ of $A$ satisfies
\[
|1-\mu+\beta_1\lambda|
\leq 1-\mu+\beta_1\rho(A).
\]
Hence $\rho(J)=1-\mu+\beta_1\rho(A)$. The strict stability and
instability conclusions follow from whether this spectral radius
is below or above one.
\end{proof}

This conclusion is specific to the mean-hazard closure, and it is not an exact threshold theorem for the finite stochastic process.
 Other closures give different linear coefficients. For instance, if the neighbors' states are treated as independent Bernoulli variables and only pairwise transmission is kept, the infection probability becomes
\[
1-\prod_j\left[1-x_j(1-e^{-\beta_1 A_{ij}})\right].
\]
For binary $A$, its linear coefficient is $(1-e^{-\beta_1})A$, rather than $\beta_1A$. These two coefficients are only consistent under the first-order approximation of the single-contact hazard. Both independent-state closures place activation requiring at least two infectious sources beyond first order in prevalence, but neither establishes bistability or an exact finite-network epidemic threshold.

\subsection{Pairwise projections do not identify higher-order dynamics}
\label{subsec:nonidentifiability}
Pairwise projections do not in general determine the higher-order contagion
dynamics, even when pairwise multiplicities are retained. The following
example gives an explicit construction.

\begin{proposition}
There exist two hypergraphs with identical weighted pairwise projections
that produce different one-step HyCo-SIS transition probabilities from the
same labeled initial condition.
\end{proposition}

\begin{proof}
Let
\[
V:=\{1,2,3,4,5,6\}
\]
and consider the two triangle hypergraphs
\[
H_A
:=
\{
\{1,2,3\},
\{1,4,5\},
\{2,4,6\},
\{3,5,6\}
\}
\]
and
\[
H_B
:=
\{
\{1,2,4\},
\{1,3,5\},
\{2,3,6\},
\{4,5,6\}
\}.
\]

For every pair of vertices \(\{i,j\}\), the number of hyperedges containing
both \(i\) and \(j\) is the same in \(H_A\) and \(H_B\). Hence the two
hypergraphs have identical weighted pairwise projections.

Now take
\[
r=2,\qquad
\beta_1=0,\qquad
\gamma=0,\qquad \beta_h>0,
\]
and suppose that initially vertices
\[
\{1,2,3\}
\]
are infected while vertices \(4,5,6\) are susceptible.

In \(H_A\), none of the susceptible vertices \(4,5,6\) belongs to a triangle
containing two infected vertices. Therefore
\[
\lambda_4=\lambda_5=\lambda_6=0,
\]
and all three remain uninfected through the higher-order channel with
probability one during that step.

In \(H_B\), vertex \(4\) belongs to the triangle \(\{1,2,4\}\), vertex \(5\)
belongs to \(\{1,3,5\}\), and vertex \(6\) belongs to \(\{2,3,6\}\). Each
susceptible vertex therefore belongs to an activated triangle containing two
infected vertices. Consequently,
\[
\lambda_4=\lambda_5=\lambda_6=\beta_h,
\]
so each is infected through the higher-order channel with probability
\[
1-\exp(-\beta_h).
\]

Thus the two hypergraphs have the same weighted pairwise projection and the
same labeled initial condition but different HyCo-SIS transition
probabilities.
\end{proof}

The proposition shows that even a weighted pairwise projection need not determine higher-order transition probabilities from a fixed labeled initial condition. The missing information is how pairwise co-occurrences are grouped into higher-order interactions.

This construction gives a fixed-label illustration of the projection ambiguity studied by Salova and D'Souza in the context of cluster synchronization \cite{salova2021projection}. The two hypergraphs are isomorphic: interchanging vertices $3$ and $4$ maps $H_A$ to $H_B$. Under uniform parameters and a relabeling-invariant initialization law, observables invariant under vertex relabeling therefore have identical distributions on the two hypergraphs. The proposition instead compares the same fixed labeled initial condition, with vertices $1,2,3$ initially infected, on both hypergraphs.

No two distinct triangles share an edge in either hypergraph, so $O_i^{(a)}(t)=0$ for every target, pathogen, and carrier configuration. The example concerns information lost about the triangle transmission channel; it does not isolate an effect of overlap reinforcement.

\section{Numerical Results}
\label{sec:numerics}
This section describes finite-time simulations and specifies which graph ensemble each one uses. Outcomes include every trial, including extinction and low prevalence. The  experiments in this section set $\alpha=1$: infection with the other pathogen does not change susceptibility to a second disease, although shared network connection topology and initial settings can still affect aggregated co-infection statistics.

\begin{table*}[!t]\centering\small
\begin{tabular}{ll}
Study & Protocol\\\hline
\parbox[t]{.21\textwidth}{\raggedright Seed dependence\strut} & \parbox[t]{.69\textwidth}{\raggedright Ten graphs, thirty trials per graph and seed fraction; $q=.005,.01,.02,.03,.04$ per pathogen; $\gamma=0$.\strut} \\[3pt]
\parbox[t]{.21\textwidth}{\raggedright Overlap reinforcement\strut} & \parbox[t]{.69\textwidth}{\raggedright Ten graphs per reuse value $0,.25,.5$; thirty trials per graph and gamma; $\gamma=0,.05,.1,.2,.4$; $q=.02$.\strut} \\[3pt]
\parbox[t]{.21\textwidth}{\raggedright Projection evaluation\strut} & \parbox[t]{.69\textwidth}{\raggedright Four training and six evaluation graphs at reuse .25, reference $\gamma=.4$; thirty matched trials per graph; rates fitted only to training endpoint union prevalence.\strut} \\[3pt]
\parbox[t]{.21\textwidth}{\raggedright Shared dynamics\strut} & \parbox[t]{.69\textwidth}{\raggedright $n=800$, 1,600 triangles; $\beta_1=.005$, $\beta_h=.1$, $\mu=.1$, $\alpha=1$, $r=2$ for HyCo-SIS. Synchronous updates; one dimensionless step; disjoint initial pathogen sets.\strut} \\[3pt]
\parbox[t]{.21\textwidth}{\raggedright Outcomes\strut} & \parbox[t]{.69\textwidth}{\raggedright Endpoint $T=200$; terminal-window states 151--200; outbreak if union prevalence $\geq.05$; extinction if no pathogen remains; other positive runs are low prevalence.\strut} \\[3pt]
\parbox[t]{.21\textwidth}{\raggedright Replication\strut} & \parbox[t]{.69\textwidth}{\raggedright Trials averaged within graph for unconditional outcomes; graph-block bootstrap intervals. Fixed graph/initialization/simulation identifiers retained; gamma comparisons use matched initial sets and simulation seeds.\strut} \\[3pt]
\end{tabular}
\caption{Simulation protocol for the included experiments. Parameters are per-step integrated hazards; $\mu$ is a recovery probability. Exact graph snapshots and per-run identifiers are available in the code repository.}\label{tab:protocol}
\end{table*}
\begin{figure*}[!t]
    \centering
    \includegraphics[width=\textwidth]{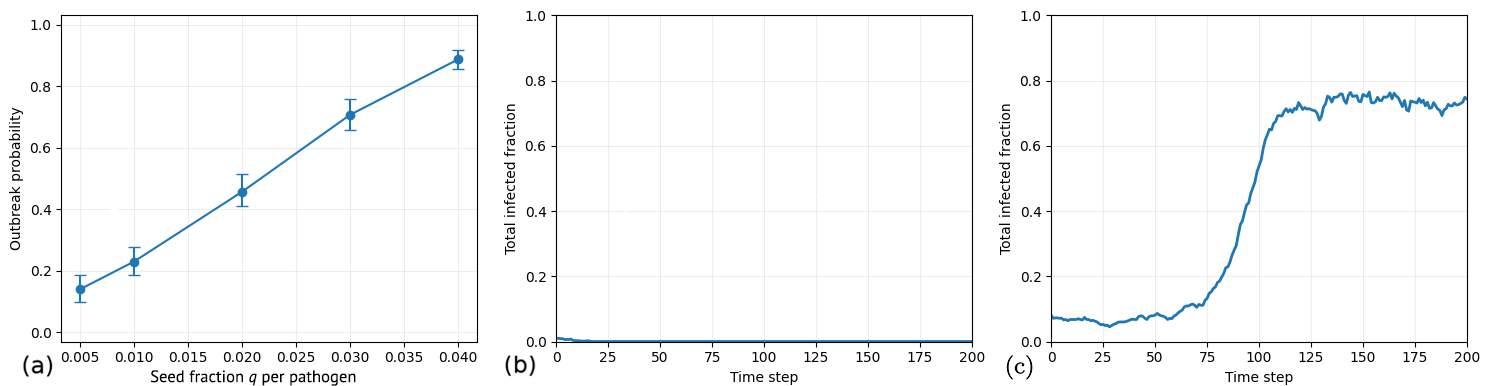}
    \caption{\textbf{Finite-seed dependence.} (a) Outbreak probability against the fraction $q$ seeded with each pathogen, averaged over $10$ hypergraphs with $30$ trials each; error bars are $95\%$ bootstrap intervals across hypergraphs. (b) A representative extinction at $q=0.005$ and (c) a representative outbreak at $q=0.04$, each the run closest to the median final prevalence of its class. Parameters as in Table \ref{tab:protocol}, with $\gamma=0$.}
    \label{fig:critical-mass}
\end{figure*}
\subsection{Synthetic hypergraph generation}

All numerical experiments are conducted on synthetic hypergraphs generated using hyperedge lists and incidence-based data structures in Python. The main synthetic family used in the experiments is a simplicial-style hypergraph generated from triangles. When a triangle \(\{a,b,c\}\) is added, the pairwise edges \(\{a,b\}\), \(\{a,c\}\), and \(\{b,c\}\) are also included. The final generated hypergraph retains both types of relationships in full: it includes pairwise contact relationships between only two nodes, as well as genuine higher-order interactions with the joint participation of multiple agents.

We chose the triangle density so that almost all vertices lie in a single connected component; sparser choices, such as $200$ triangles on $800$ vertices, fragment the hypergraph into many small components. The principal baseline used in the numerical experiments has
\[
n=800,\qquad |E_\Delta|=1600.
\]

We generated $50$ such hypergraphs independently. In each, the largest connected component contains at least $99\%$ of the vertices (on average $99.77\%$), and $8$ of the $50$ are fully connected. These $50$ hypergraphs are used only to choose the density; all dynamics are simulated on separately generated graphs, described below. All structural statistics for this set of baseline hypergraphs are summarized in Table \ref{tab:baseline-structure}.

\begin{table}[H]
\centering
\begin{tabular}{lc}
\hline
Quantity & Value \\
\hline
Number of vertices $n$ & $800$ \\
Number of triangles $|E_\Delta|$ & 1,600 \\
Mean pairwise degree $\langle k\rangle$ & $11.91$ \\
Giant-component fraction & $0.9977$ \\
Spectral radius $\rho(A)$ & $13.97$ \\
Mean triangle degree $\langle d_\Delta\rangle$ & $6.00$ \\
Mean overlap opportunity $\langle\Omega_i\rangle$ & $0.0873$ \\
Isolated-node fraction & $0.00235$ \\
\hline
\end{tabular}
\caption{Structural statistics of the principal random simplicial contact ensemble, averaged over $50$ independently generated hypergraphs.}
\label{tab:baseline-structure}
\end{table}

\paragraph{Edge-reuse construction.}
For the overlap ensemble, triangles are added until 1,600 distinct triples are obtained. If at least one triangle-side pair has already been generated, each attempt reuses such a pair with probability $p_{\rm reuse}$. A reused pair is selected uniformly from the distinct used pairs; the third vertex is sampled uniformly from the other vertices. Otherwise three distinct vertices are sampled uniformly. Duplicate triangles are rejected and a new attempt is made. After all triangle sides are collected, uniformly sampled missing pairwise contacts are added until the skeleton has 4,800 edges. The density-calibration ensemble does not impose this same edge-count constraint.  

For distinct triangles, structural overlap opportunity is
\[
\Omega_i:=\sum_j {M_{ij}\choose2},
\]
where $M_{ij}$ counts incident triangles containing pair $\{i,j\}$. This counts pairs of triangles sharing an additional source vertex at target $i$. It is an absolute count, not a normalized frequency. Matching mean degree alone does not fix this quantity or other structural properties. 

Table \ref{tab:actual-graphs} summarizes the thirty graphs used in the overlap experiments, ten for each reuse value. By construction they share the same numbers of vertices, triangles and pairwise contacts, so the mean degree is fixed at $12$, while degree heterogeneity, spectral radius and overlap opportunity all grow with $p_{\rm reuse}$. The simulation protocol for all experiments is summarized in Table \ref{tab:protocol}.

\begin{table}[h]\centering\small
\begin{tabular}{lrrrrrr}
Reuse & Graphs & Mean GCC & $\langle k\rangle$ & Degree SD & $\rho(A)$ & $\langle\Omega\rangle$\\\hline
0 & 10 & 0.9976 & 12.00 & 4.82 & 14.07 & 0.0887 \\
0.25 & 10 & 0.9984 & 12.00 & 5.05 & 14.29 & 1.2095 \\
0.5 & 10 & 0.9975 & 12.00 & 5.60 & 14.87 & 2.4502 \\
\end{tabular}
\caption{Structural statistics of the thirty overlap-experiment graphs, averaged over the ten graphs at each reuse value. Every graph has 800 vertices, 1,600 triangles, mean triangle degree six, and 4,800 distinct pairwise contacts. Mean GCC is the fraction of vertices in the largest connected component, and Degree SD is the standard deviation of pairwise degree within each graph. The projection experiment uses the reuse-0.25 graphs, split four/six.}\label{tab:actual-graphs}
\end{table}

\subsection{Finite-seed dependence}
\label{subsec:seed}
To investigate finite-seed effects on a well-connected higher-order contact structure, we performed two-pathogen HyCo-SIS simulations on hypergraphs at the baseline density,
\[
n=800,\qquad |E_\Delta|=1600.
\]

The contagion parameters were fixed at
\[
\beta_1=0.005,\qquad
\beta_h=0.10,\qquad
\mu=0.10,
\]
with
\[
\alpha=1,\qquad
\gamma=0,\qquad
r=2.
\]

For each initial seed fraction
\[
q\in\{0.005,0.01,0.02,0.03,0.04\},
\]
a disjoint fraction $q$ of the population was initially infected with pathogen 1 and another disjoint fraction $q$ with pathogen 2. The seed study uses ten independently generated hypergraphs of this size, whose exact snapshots are available in our code repository (see Code and data), with $30$ stochastic contagion trials per graph and seed fraction, for a total of $300$ trials per seed size.

An outbreak was defined as a realization whose total infected fraction at step 200 was at least $0.05$. Outbreak probabilities were first computed separately for each hypergraph and then averaged across the $10$ independent graph realizations.

Figure \ref{fig:critical-mass} shows finite-time seed dependence:
mean outbreak probability increases from approximately $0.14$ at
$q=0.005$ to $0.887$ at $q=0.04$. These outcomes are measured at
step $200$ and demonstrate sensitivity to introduction size.
The spectral statistics provide context for this dependence.
At $\beta_1=0.005$ and $\mu=0.10$, the mean spectral radius
$13.97$ of the separate structural-calibration ensemble in
Table \ref{tab:baseline-structure} gives
$\beta_1\rho(A)/\mu\approx0.699$.
The corresponding ensemble-mean ratios for the overlap graphs
in Table \ref{tab:actual-graphs} are approximately $0.704$,
$0.714$, and $0.744$ at reuse $0$, $0.25$, and $0.5$,
respectively. These means lie below the stability boundary of
Proposition \ref{prop:dfe}; they are not graphwise upper bounds
or measurements of the ten seed-study graphs.
The observed seed dependence is consistent with a role for
nonlinear group activation, which is absent from the disease-free
Jacobian. It does not by itself establish bistability, asymptotic
persistence, or the contribution of each transmission channel
to individual outbreaks.

\subsection{Calibrated pairwise projections}
\label{subsec:projections}
We compare HyCo-SIS with the clique ($C$), weighted ($W$), and
threshold ($Q$) projections on the edge-reuse ensemble with
$p_{\rm reuse}=0.25$. Each graph has $n=800$ vertices,
1,600 triangles, and 4,800 distinct pairwise contacts.
The HyCo-SIS reference uses $\beta_1=0.005$, $\beta_h=0.1$,
$\gamma=0.4$, $\mu=0.1$, $\alpha=1$, and $r=2$, with disjoint
initial fractions of $0.02$ per pathogen.

Four graphs are used
for calibration and six different graphs for evaluation, with
$30$ matched trials per graph. For each projection, the pairwise
transmission coefficient $\widehat\beta$ is selected from a
prespecified search grid to minimize the absolute discrepancy
between its mean endpoint union prevalence and the HyCo-SIS
training mean at step $200$. Co-infection is not used for fitting.
The fitted coefficients are then held fixed on the evaluation
graphs. Figure \ref{fig:calibrated-projection-evaluation}
summarizes the calibration and held-out errors, and
Fig. \ref{fig:calibrated-projection-trajectories} shows the
corresponding unconditional mean trajectories.

The fitted coefficients are $\widehat\beta_C=0.0170$,
$\widehat\beta_W=0.0085$, and $\widehat\beta_Q=0.0190$.
The HyCo-SIS training target is $69.7\%$, compared with fitted
training means of $69.9\%$, $69.4\%$, and $70.1\%$ for
$C$, $W$, and $Q$, respectively. On the six held-out graphs,
the projections give mean endpoint union prevalence between
$69.4\%$ and $69.8\%$, compared with $67.8\%$ for HyCo-SIS.
Calibration is therefore approximate on the training graphs
and does not impose agreement on unseen graphs.

Mean endpoint co-infection for HyCo-SIS on the six held-out
graphs is $19.7\%$. The mean signed errors, defined as
projection minus HyCo-SIS, are $3.17$, $3.17$, and $3.79$
percentage points for $C$, $W$, and $Q$, respectively.
Each projection overestimates endpoint co-infection on five
evaluation graphs and underestimates it on one. All three
$95\%$ intervals for the mean signed error include zero,
while the mean absolute graph-level errors are $6.68$,
$6.64$, and $7.03$ percentage points. Thus the endpoint
results show graph-level discrepancies but do not establish
a positive mean signed bias or a ranking of the projections.

For co-infection averaged over states $151$--$200$, the mean
signed errors are $5.79$ $[0.57,9.55]$, $5.84$ $[0.55,9.61]$,
and $6.39$ $[1.16,10.19]$ percentage points for $C$, $W$, and
$Q$, respectively. Brackets denote $95\%$ percentile bootstrap
intervals obtained from 20,000 resamples of the six paired
graph means. Every held-out projection trial satisfies the
$5\%$ endpoint outbreak criterion, compared with an outbreak
frequency of $85.6\%$ for HyCo-SIS. Together with the earlier
spread in Fig. \ref{fig:calibrated-projection-trajectories},
these results show that fitting endpoint union prevalence
does not ensure agreement in transient dynamics, terminal-window
co-infection, or outbreak frequency. The conclusions concern
the three fitted models on this six-graph evaluation ensemble
at a finite horizon; they do not establish equilibrium behavior
or a general impossibility of pairwise approximation.

\begin{figure*}[!t]
\centering
\includegraphics[width=.86\textwidth]{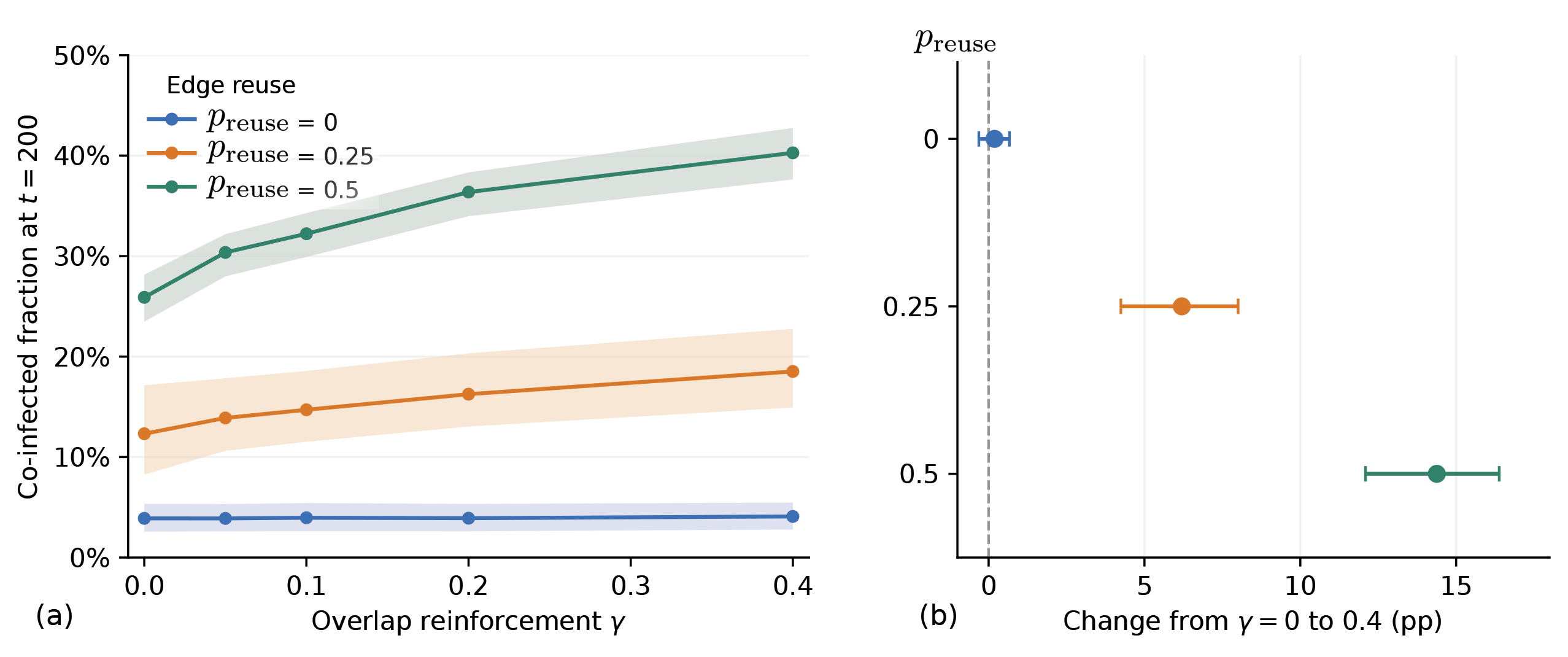}
\caption{Effect of overlap reinforcement on co-infection at
$t=200$. (a) Mean co-infected fraction versus $\gamma$.
(b) Paired changes from $\gamma=0$ to $\gamma=0.4$.
Each reuse setting includes ten graphs and thirty trials per
graph and reinforcement value, including extinct and
low-prevalence runs. Bands and intervals are pointwise $95\%$
graph-bootstrap intervals from 20,000 draws. The paired increase is positive at reuse $0.25$
and $0.5$; its interval includes zero at reuse $0$.}
\label{fig:overnight-overlap}
\end{figure*}
\begin{figure*}[!t]
  \centering
  \includegraphics[width=.86\textwidth]{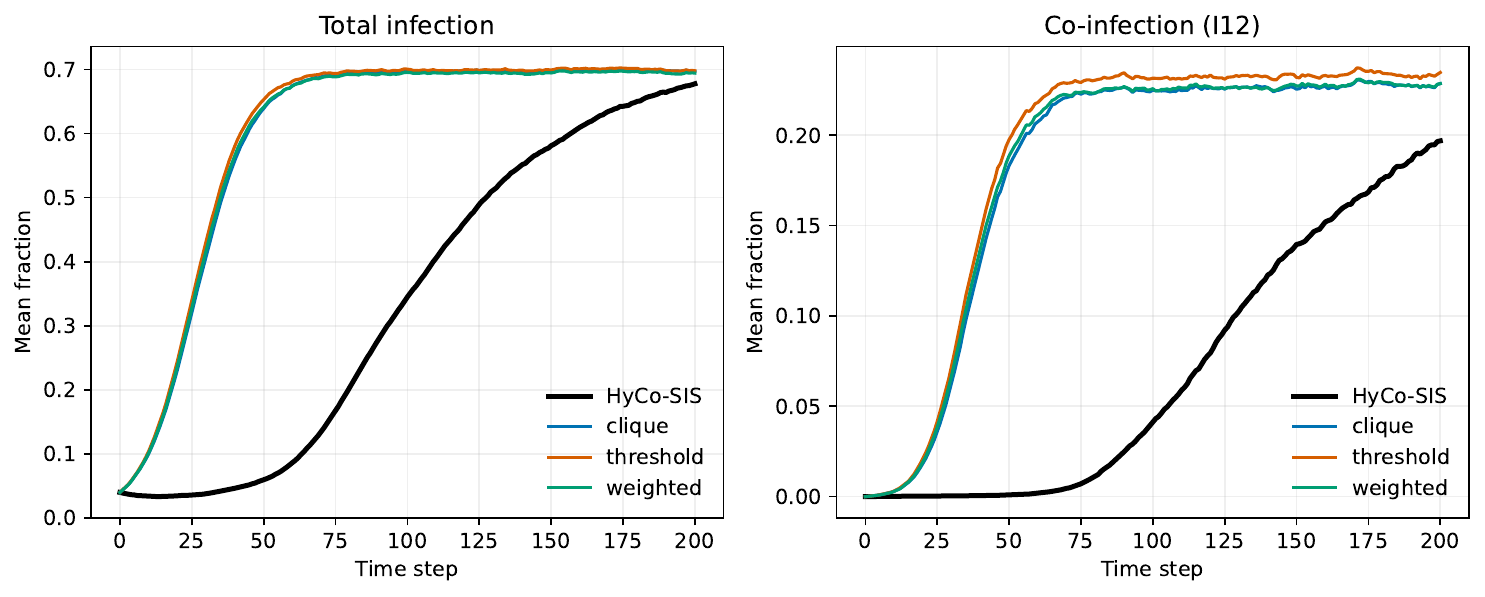}
 \caption{Mean held-out trajectories of total infection (left) and co-infection (right) for HyCo-SIS and the three fitted projections, averaged over $30$ trials on each of six held-out graphs, including extinctions and low-prevalence runs.}
  \label{fig:calibrated-projection-trajectories}
\end{figure*}
\subsection{Overlap reinforcement on the saved graph ensemble}
\label{subsec:overlap-results}
We analyzed 4,500 trajectories on the 30 graphs of Table \ref{tab:actual-graphs}: ten graphs for each edge-reuse value $p_{\rm reuse}\in\{0,0.25,0.5\}$, five reinforcement values $\gamma\in\{0,0.05,0.1,0.2,0.4\}$, and 30 trials per graph and reinforcement value. Each graph has 800 vertices, 1,600 triangles, and 4,800 distinct pairwise contacts. Parameters are $\beta_1=0.005$, $\beta_h=0.10$, $\mu=0.10$, $\alpha=1$, $r=2$, and 200 synchronous updates. Initial sets are disjoint with 2\% carriers per pathogen. All results include extinct and positive low-prevalence trials.

We averaged trials within graphs and paired graph means across
reinforcement values before obtaining $95\%$ percentile intervals
from 20,000 graph-bootstrap draws. As shown in
Fig. \ref{fig:overnight-overlap}, increasing $\gamma$ from $0$
to $0.4$ increased endpoint co-infection by $6.20$ percentage
points at reuse $0.25$ and $14.38$ points at reuse $0.5$.
At reuse $0$, the estimated $0.19$-point change was small and
its interval included zero (Table \ref{tab:overnight-paired}).
Appendix \ref{app:overlap} reports the corresponding infectious-overlap
trajectories for all targets and for targets lacking the relevant
pathogen. The paired reinforcement contrasts hold topology fixed;
cross-reuse differences also reflect changes in degree
heterogeneity and spectral radius.

\begin{table}[h]
\centering\small
\setlength{\tabcolsep}{3pt}
\begin{tabular}{rrrr}
Reuse &
\shortstack{$\bar\rho_{12}^{\gamma=0}$\\(\%)} &
\shortstack{$\bar\rho_{12}^{\gamma=0.4}$\\(\%)} &
\shortstack{Change [95\% CI]\\(pp)}\\\hline
0 & 3.90 & 4.09 & 0.19 [$-0.31$, 0.67] \\
0.25 & 12.32 & 18.52 & 6.20 [4.24, 8.01] \\
0.5 & 25.90 & 40.28 & 14.38 [12.09, 16.38] \\
\end{tabular}
\caption{Mean endpoint co-infection at $t=200$ and the paired
change from $\gamma=0$ to $\gamma=0.4$, in percentage points
(pp). Each reuse setting includes ten graphs and thirty trials
per graph and reinforcement value. Intervals are $95\%$
percentile bootstrap intervals obtained by resampling paired
graph means; all trials are included.}
\label{tab:overnight-paired}
\end{table}

\paragraph{Outcome definitions and finite-time interpretation.}
Outbreak means union prevalence at least 0.05 at $t=200$; extinction means zero carriers; the remaining positive runs are low prevalence. The full grid contains 320 positive non-outbreak trajectories. At reuse 0.5 and $\gamma=0.4$, outbreak probability is $0.9667$, indicating saturation of this binary outcome rather than complete infection. Endpoint minus mean co-infection over states 151--200 remains positive in all three reuse regimes at $\gamma=0.4$; longer horizons are needed before claiming equilibrium.

\subsection{Outbreak frequency and conditional co-infection}
\label{subsec:conditional}
To separate changes in outbreak frequency from changes among outbreak runs, we reanalyzed the saved trajectories using the endpoint rule $\mathcal{B}:=\{\rho_{\mathrm{union}}(200)\geq0.05\}$. The conditional mean is the sum of endpoint co-infection over outbreak runs divided by the number of outbreak runs. Equal trial counts give equal graph weights before conditioning, but graphs with more outbreaks contribute more observations to this conditional mean. We resample ten paired graph blocks 20,000 times, recomputing both ratios and their difference in each draw. This is a retrospective descriptive analysis of the original grid, not a new experiment.

\begin{figure}[t!]
  \centering
  \includegraphics[width=.45\textwidth]{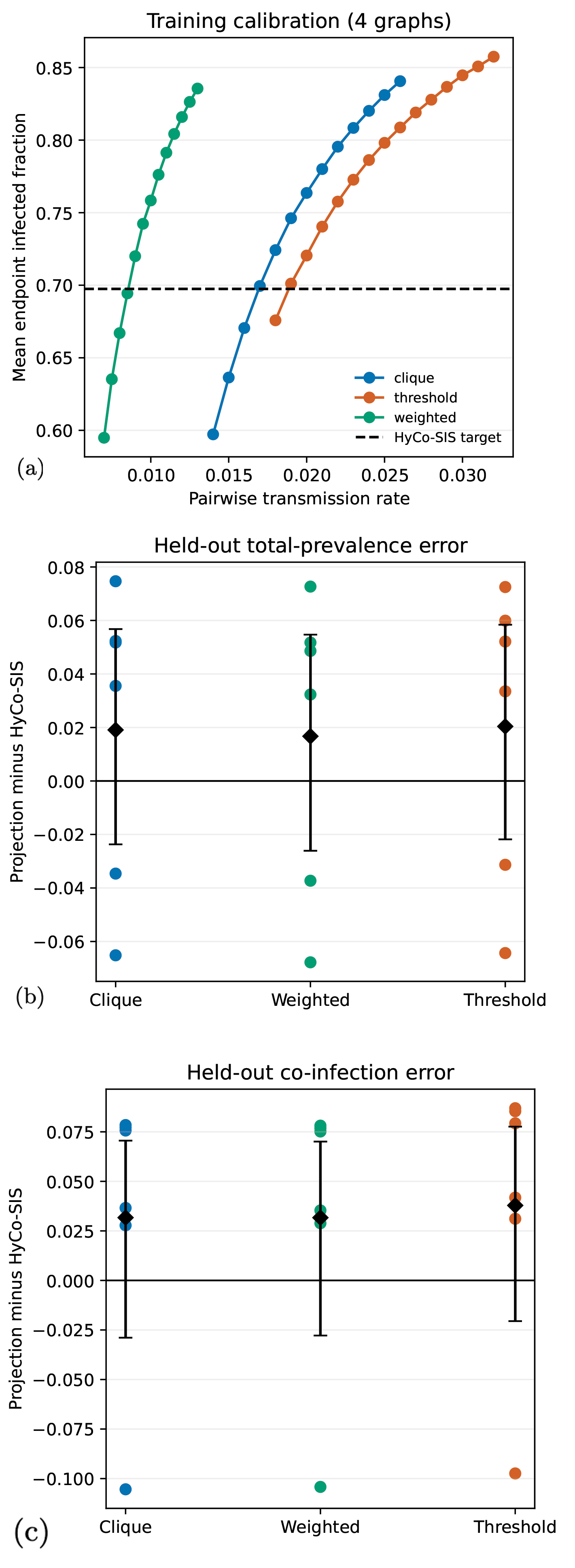}
 \caption{Calibration and held-out errors for the three pairwise projections. (a): calibration on four training graphs. (b) and (c): errors (projection minus HyCo-SIS) in total prevalence and co-infection at step $200$ on six held-out graphs. Points are graph-level errors; diamonds and bars are means and $95\%$ bootstrap intervals.}
  \label{fig:calibrated-projection-evaluation}
\end{figure}

The conditional comparison asks a different question from the unconditional mean. Because changing reinforcement also changes which runs qualify as outbreaks, this is not an estimate of a causal effect among an identical set of successful outbreaks. Non-outbreak runs may retain a small positive prevalence and are not automatically extinct. Define $Y:=\rho_{12}(200)$. Writing $E$ and $P$ for
expectation and probability over the trial-and-graph
sampling described above, the exact identity is
\[
E[Y]
=P(\mathcal{B})E[Y\mid\mathcal{B}]
+E[Y\mathbf1_{\mathcal{B}^c}],
\]
where $\mathcal{B}^c$ denotes the complement of
$\mathcal{B}$. Sensitivity to endpoint cutoffs of 1\% and 10\% is recorded alongside the primary 5\% analysis; no cutoff was selected to maximize an effect. At both sensitivity cutoffs, the conditional-change intervals remain positive for reuse 0.25 and 0.5 and include zero for reuse zero.

At intermediate and high reuse, reinforcement increases both outbreak frequency and the conditional co-infection mean (Table \ref{tab:conditional}). At zero reuse, the conditional-change interval contains zero. In this dataset, the non-outbreak co-infection contribution is zero at both endpoint reinforcement settings under the 5\% rule. These observations help distinguish outbreak-frequency changes from changes within the selected outbreak samples.

 \begin{table}[!t]\small
\centering\small
\setlength{\tabcolsep}{2.5pt}
\begin{tabular}{rrrr}
Reuse &
\shortstack{Outbreak\\frequency (\%)} &
\shortstack{Conditional\\co-infection (\%)} &
\shortstack{Conditional change\\{[95\% CI]} (pp)}\\\hline
0 & 52.67 $\to$ 55.67 & 7.41 $\to$ 7.35 & $-0.06$ [$-0.96$,\,0.79] \\
0.25 & 79.00 $\to$ 86.67 & 15.60 $\to$ 21.37 & 5.77 [3.04,\,8.10] \\
0.5 & 92.67 $\to$ 96.67 & 27.95 $\to$ 41.67 & 13.72 [11.21,\,15.96] \\
\end{tabular}
\caption{Changes from $\gamma=0$ to $0.4$ on the same networks; ten graphs and thirty trials per graph and gamma. Conditional means use only runs meeting the endpoint outbreak rule. }\label{tab:conditional}
\end{table}
\section{Overlap reinforcement and allocation among source groups}\label{sec:allocation}
The previous section varied the strength of overlap reinforcement. Here we keep its strength fixed and ask instead how it is allocated among source groups. We first write the overlap term exactly as a sum of four-person interactions, which makes this allocation explicit. We then compare the original weights with two controlled redistributions, examine their difference on identical network states, and close with an exact five-person calculation
illustrating a mechanism consistent with the observed
acquisition differences.
\subsection{Exact four-person representation}
\label{subsec:four-person}
Consider distinct triangle hyperedges with threshold $r=2$ and binary carrier indicators $X_i$ for one pathogen.
Define
$
a_{ijk}:=\mathbf{1}\{\{i,j,k\}\in E_\Delta\},
$
where $\mathbf{1}\{\cdot\}$ denotes the indicator function.
For a four-vertex set $Q:=\{i,j,k,l\}$ (unrelated to the projection matrices $W$ and $Q$ of Section \ref{sec:model}) define
\[
 W_{i,Q}:=a_{ijk}a_{ijl}+a_{ijk}a_{ikl}+a_{ijl}a_{ikl}.
\]
All sums over $Q$ below run over four-element subsets
of $V$; groups without a contributing triangle pair
have zero original coefficients. The overlap count satisfies the exact identity
\[
 O_i(X)=\sum_{Q\ni i}W_{i,Q}\prod_{j\in Q\setminus\{i\}}X_j.
\]
Indeed, two distinct incident triangles contributing to overlap share one source vertex as well as the target; their activation product reduces using $X_j^2=X_j$.
Grouping these products by their four-vertex union gives the identity without a mean-field approximation.
The coefficients distinguish targets even though the original triangles are undirected.

For zero or one triangle face in $Q$, all four coefficients vanish. Two faces give $(1,1,0,0)$, three give $(3,1,1,1)$, and four give $(3,3,3,3)$, up to permutation.

\begin{proposition}\label{prop:symmetric}
There are constants $c_Q\geq0$ such that $O_i(X)=\sum_{Q\ni i}c_Q\prod_{j\in Q\setminus\{i\}}X_j$ for every target $i$ and every carrier state $X$ if and only if no four-set contains exactly two or three triangle faces.
\end{proposition}
\begin{proof}
For a fixed target $i$, the monomials $\prod_{j\in Q\setminus\{i\}}X_j$ with $Q\ni i$ are distinct multilinear monomials, so their coefficients are determined by $O_i$ as a function on binary states. Hence such constants exist if and only if $W_{i,Q}=c_Q$ for every $i\in Q$, that is, if and only if every four-set has equal coefficients at its four vertices. By the list above, this holds for zero, one or four faces and fails for two or three.
\end{proof}

Geometrically, a four-set with exactly two triangle faces is composed of two triangles sharing an edge such as in the edge-reuse construction. Its overlap coefficients are nonzero only at the two shared vertices. A four-set with three faces also has unequal target coefficients. By contrast, a complete four-face tetrahedron has equal coefficients at all four targets.

This characterization determines when symmetric four-person
weights reproduce every target's infection pressure for every
carrier configuration, with the pairwise and triangle transmission
rules held fixed. When such an exact representation is impossible,
a symmetric model might still approximate the overall infected
fraction; the characterization does not address that question.

SIS processes on hypergraphs with indirect spreading pathways have been studied by Cui et al. \cite{cui2024general}. The condition excluding four-sets with exactly two or three triangle faces also appears in the study of hypergraph lines \cite{beaudou2013lines}; Proposition \ref{prop:symmetric} connects it to the existence of an exact symmetric four-person representation of the HyCo-SIS overlap rule.

\subsection{Two controlled redistributions}
Define $\overline W_Q:=\frac14\sum_{i\in Q}W_{i,Q}$.
The symmetric control assigns $\overline W_Q$ to all four
targets in group $Q$. It preserves each group's total overlap
coefficient, but generally changes each target's total coefficient
across its incident groups. Group-strength-preserving symmetrization has precedent in synchronization research \cite{directedsync2022}.
To match those strengths, define
\[
 s_i:=\sum_{Q\ni i}W_{i,Q},\qquad t_i:=\sum_{Q\ni i}\overline W_Q,
\]
\[
V_{i,Q}:=\begin{cases}\overline W_Qs_i/t_i,&t_i>0,\\0,&t_i=0.\end{cases}
\]
Here $s_i$ is target $i$'s original total overlap coefficient,
and $t_i$ is its total under the symmetric control. When $t_i>0$,
multiplication by $s_i/t_i$ restores the original total.
When $t_i=0$, nonnegativity of the original coefficients implies
$s_i=0$, so assigning zero also preserves the total. Thus
$\sum_{Q\ni i}V_{i,Q}=s_i$ for every target.
This second control generally has target-dependent weights and changes individual group totals.
It therefore controls target strength without isolating directionality alone.

Under independent Bernoulli($\pi$) carrier indicators,
where $\pi\in[0,1]$ is the common carriage probability, the original and strength-matched rules have identical expected overlap hazard at each target, $\gamma\pi^3s_i$, and identical expected eligible-target overlap hazard, $\gamma(1-\pi)\pi^3s_i$.
Equality of these expectations does not imply equality of acquisition probabilities $1-\exp(-\lambda_i)$.

There is a simple obstruction to simultaneously retaining symmetric nonnegative weights, existing group support, and original node strengths.
A zero-strength vertex forces every incident symmetric group weight to be zero.
If every group containing a positive-strength vertex also contains a zero-strength vertex, its positive strength cannot be recovered.
All ten networks in the following comparison admit this explicit certificate.
This is a support-specific obstruction, not an impossibility for representations with additional groups.

\subsection{Finite-time comparisons and coupling robustness}
\label{subsec:allocation-results}
We used ten saved networks at edge reuse $0.25$, with thirty matched initialization and simulation seeds per network and model.
Each network has $800$ vertices, 1,600 triangles and 4,800 distinct pairwise contacts.
Parameters were $\beta_1=0.005$, $\beta_h=0.1$, $\gamma=0.4$,
recovery probability $\mu=0.1$, threshold $r=2$, disjoint initial
pathogen fractions $0.02$, and $T=200$ synchronous steps.
All other event rules were identical.
Trials were averaged within networks; descriptive $95\%$ percentile intervals used 20,000 paired graph-block bootstrap resamples. 

\begin{table}[h!]
\centering\small
\begin{tabular}{lrr}
\hline
Model & $\alpha=1$ & $\alpha=2$\\
\hline
Original overlap & $18.52\%$ & $39.70\%$\\
Symmetric control & $32.51\%$ & $53.90\%$\\
Target-strength matched & $28.54\%$ & $47.79\%$\\
\hline
\end{tabular}
\caption{Mean co-infection at step $200$ for the original overlap rule and its two controls, on ten saved networks at edge reuse $0.25$ with $\gamma=0.4$.}
\label{tab:allocation}
\end{table}
The exploratory intervals are not adjusted for multiple comparisons. Table \ref{tab:allocation} lists the mean co-infection at step $200$, and Fig. \ref{fig:allocation-alpha} shows the values for each of the ten networks.
The strength-matched minus original difference was $10.02$ percentage points at $\alpha=1$, with interval $[6.23,13.97]$, and $8.09$ points at $\alpha=2$, with interval $[5.29,10.98]$.
The change in this contrast between $\alpha=1$ and $2$ was $-1.93$ points, with interval $[-6.96,2.00]$.
Thus the contrast persists with increased secondary-acquisition susceptibility, but an amplification or attenuation by that coupling is not established. The sign of the contrast is also consistent with Proposition \ref{prop:exchangeable-gain}, which holds for all nonnegative pairwise and triangle rates.
At $\alpha=1$ the model has no cross-susceptibility reinforcement; the original result alone is not an intrinsically interacting-pathogen effect.

To assess sensitivity to the rate parameters, we repeated
the allocation comparison on the same ten networks with
thirty matched trials per network, fixing $\alpha=1$ and
$\gamma=0.4$. We tested $\beta_h=0.05$ and $0.2$ at
$\mu=0.1$, and separately tested $\mu=0.05$ at
$\beta_h=0.1$, with all other settings unchanged.
Table~\ref{tab:allocation-rate-robustness} reports the results.
The mean target-strength-matched minus original contrast
remained positive in all three checks, with each bootstrap
interval above zero, although its magnitude was substantially
smaller than at baseline.  This is expected from Table \ref{tab:rate-sensitivity}: at $\beta_h=0.05$ co-infection is nearly absent under the original rule, while at $\beta_h=0.2$ and at $\mu=0.05$ it is already close to saturation, leaving little room for any redistribution to act. The allocation effect is therefore largest at intermediate prevalence, as in the baseline.

\begin{table}[!htbp]
\centering
\small
\setlength{\tabcolsep}{3pt}
\renewcommand{\arraystretch}{1.12}
\begin{tabular}{@{}lrrrr@{}}
\hline
$\beta_h$ & 0.10 & 0.05 & 0.20 & 0.10 \\
$\mu$    & 0.10 & 0.10 & 0.10 & 0.05 \\
\hline
Original (\%)
  & 18.52 & 0.0025 & 62.08 & 78.58 \\
Symmetric (\%)
  & 32.51 & 1.95 & 65.71 & 80.30 \\
Strength-matched (\%)
  & 28.54 & 0.89 & 64.33 & 78.71 \\
\hline
Matched minus original (pp)
  & 10.02 & 0.89 & 2.25 & 0.13 \\
95\% CI lower (pp)
  & 6.23 & 0.34 & 0.37 & 0.10 \\
95\% CI upper (pp)
  & 13.97 & 1.55 & 4.16 & 0.16 \\
\hline
\end{tabular}
\caption{%
Rate sensitivity of mean co-infection at step $200$,
with $\alpha=1$ and $\gamma=0.4$.
The first numeric column reproduces the baseline;
the remaining columns give the three additional checks.
All settings use the same ten reuse-$0.25$ networks
and thirty matched trials per network and model,
with all other settings as in the baseline.
All trials are included.
Intervals use $20{,}000$ paired graph-bootstrap
resamples of network means, without multiplicity
adjustment. Differences are in percentage points (pp).
}
\label{tab:allocation-rate-robustness}
\end{table}

\subsection{Identical-state diagnostics}
\label{subsec:identical-state}
We replayed five predetermined trials per network for both original and strength-matched dynamics, saving node states at times $0,50,100,150,199$.
Both hazard rules were evaluated on each identical state and after twenty random permutations of complete four-state node labels.
The permutations preserve compartment counts.
For $K$ uniformly placed carriers on $n$ vertices, a specified target is absent while three distinct specified sources are present with probability
\[
 \frac{n-K}{n}\frac{K(K-1)(K-2)}{(n-1)(n-2)(n-3)}.
\]
Consequently, equal total coefficients give exactly equal expected eligible overlap hazard under permutation.

\begin{figure}[h!]
\centering
\includegraphics[width=0.5\textwidth]{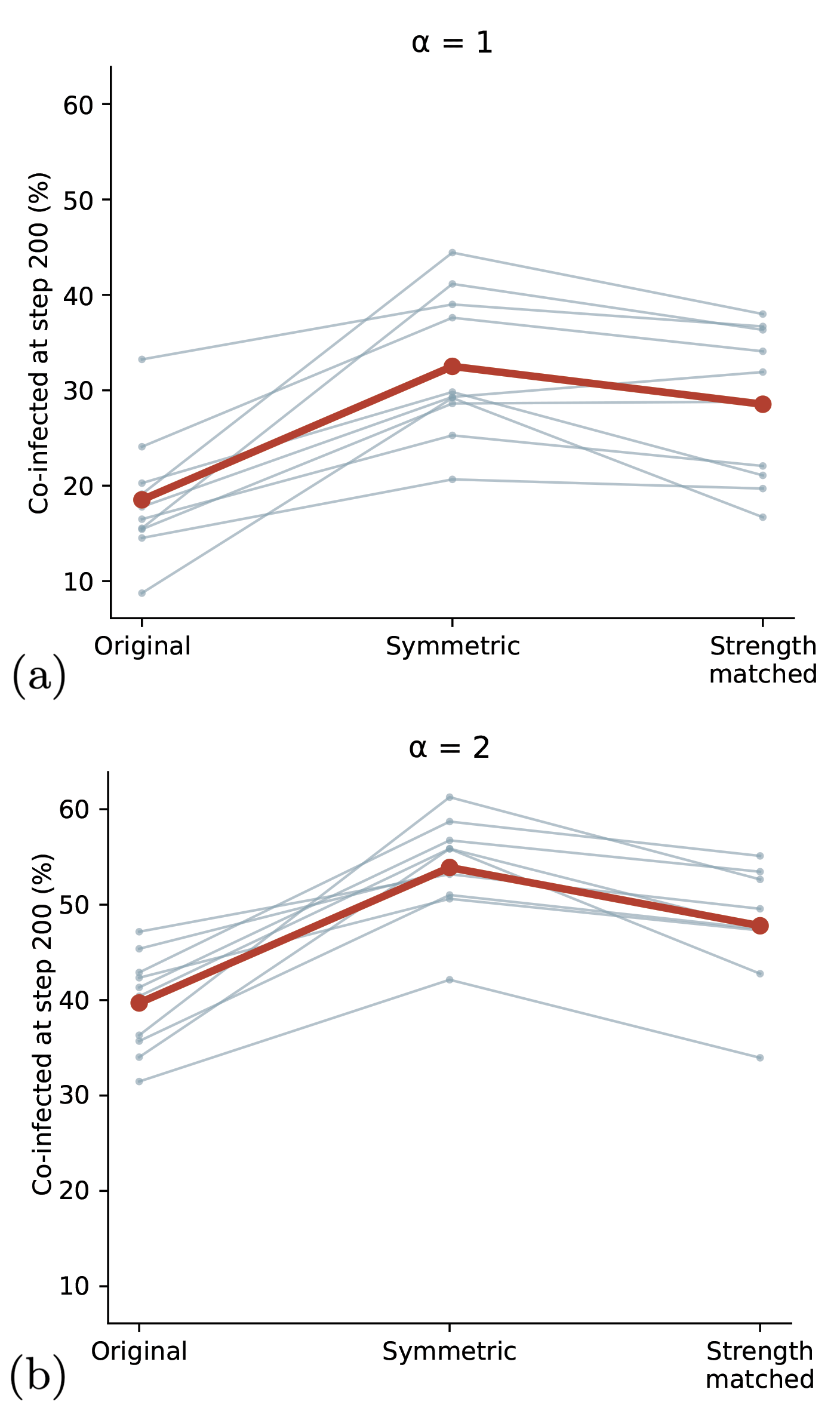}
\caption{Co-infection at step $200$ under the original overlap rule and its two controls, for (a) $\alpha=1$  and (b) $\alpha=2$,  with $\gamma=0.4$ and reuse $0.25$. Thin lines are the thirty-trial averages on each of the ten networks; the thick line is their mean.}
\label{fig:allocation-alpha}
\end{figure}
\begin{figure*}[!t]
\centering
\includegraphics[width=\textwidth]{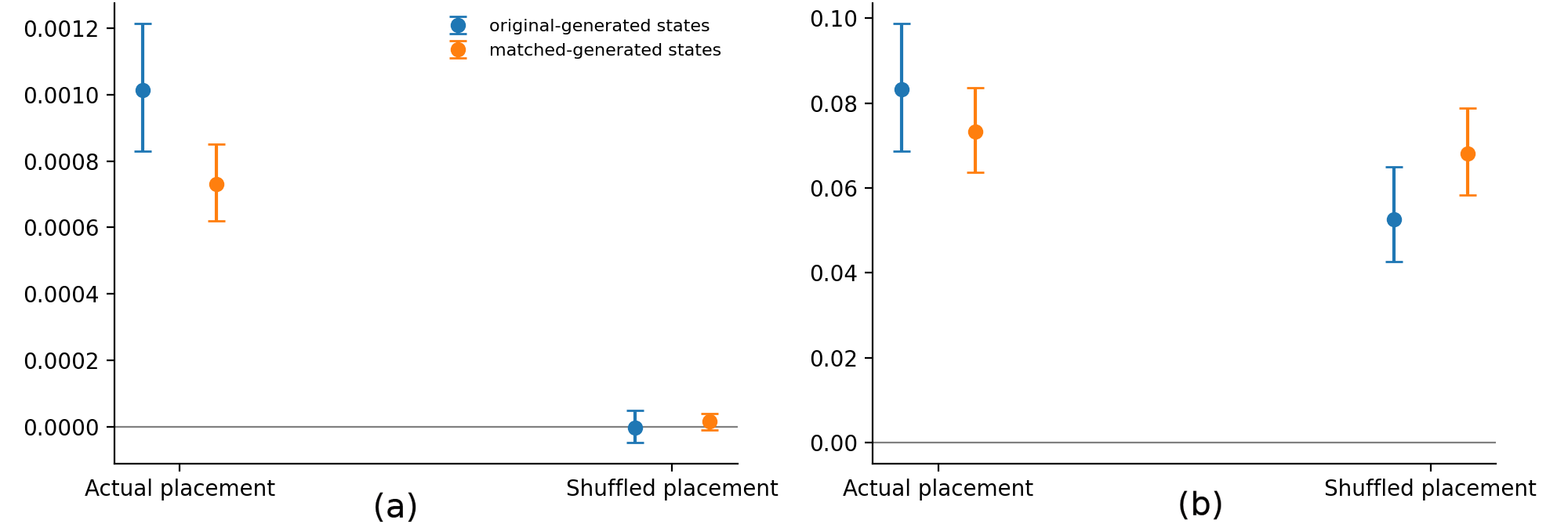}
\caption{Strength-matched minus original rule evaluated on identical states, before and after twenty random permutations of the node labels, for states generated by each rule at $\alpha=1$. (a): eligible overlap hazard. (b): acquisition probability, in percentage points per node per pathogen per step. Error bars are $95\%$ graph-bootstrap intervals.}
\label{fig:allocation-states}
\end{figure*}

On original-generated states (Fig. \ref{fig:allocation-states}), averaging times $150$ and $199$, the strength-matched minus original acquisition-probability difference was $0.0832$ percentage points per population member per pathogen per step, with interval $[0.0687,0.0987]$.
After shuffling it remained $0.0527$ points, with interval $[0.0425,0.0649]$, although the expected overlap-hazard difference vanishes.
The actual-minus-shuffled acquisition contrast was positive for this ensemble, but unresolved for the strength-matched-generated ensemble.
These diagnostics distinguish state-placement effects from differences that survive random placement; they do not decompose the endpoint contrast into causal fractions or distinguish individual prevalence heterogeneity from inter-node correlations.

\subsection{An exact illustrative example}
\label{subsec:five-person}
Return to the five students in Fig. \ref{fig:mechanism}, with triangles $012$, $013$ and $024$, and their induced pairwise edges.
For the supported groups $0123$ and $0124$, target $1$ has original coefficients $(1,0)$, replaced by $(1/2,1/2)$ under target-strength matching; target $2$ has the reversed original coefficients and the same matched coefficients.
Target $0$ retains $(1,1)$; the other targets retain zero overlap strength.

\begin{proposition}[One-step gain under exchangeable initialization]\label{prop:exchangeable-gain}
Consider the single-pathogen process on the five-vertex hypergraph with triangles $012,013,024$ and their pair-edge union, with threshold $r=2$. Write $b:=\beta_1\geq0$, $h:=\beta_h\geq0$, and $g:=\gamma\geq0$ for the pairwise, triangle, and reinforcement rates, respectively. Compare the original overlap coefficients with the target-strength-matched control described above.
Suppose the initial single-pathogen carrier configuration is exchangeable
(its distribution is unchanged by relabeling the vertices).
Let $K_0$ denote the initial number of carriers and define
$q_3:=\mathbb{P}(K_0=3)$.
Let $\rho_{\mathrm{matched}}(1)$ and
$\rho_{\mathrm{original}}(1)$ denote the single-pathogen
carrier fractions after one update under the two rules.
Define
\[
\Delta:=
\mathbb{E}[\rho_{\mathrm{matched}}(1)]
-\mathbb{E}[\rho_{\mathrm{original}}(1)].
\]
With the same recovery probability in both models,
this difference satisfies
\begin{align}
\Delta&=\frac{q_3}{25}(1-e^{-g/2})
\left(e^{-\Lambda_2}-e^{-\Lambda_1-g/2}\right),\\
\Lambda_1&:=3b+2h,\qquad \Lambda_2:=2b+h.
\end{align}
In particular, $\Delta>0$ if $q_3>0$ and $g>0$, for every choice of the pairwise and triangle rates $b,h\geq0$. For independent Bernoulli($\pi$) initial carriers,
\begin{equation}
\Delta(\pi)=\frac{2}{5}\pi^3(1-\pi)^2(1-e^{-g/2})
\left(e^{-\Lambda_2}-e^{-\Lambda_1-g/2}\right).
\end{equation}
\end{proposition}
\begin{proof}
The supports $0123$ and $0124$ have original coefficient rows
$(1,1),(1,0),(0,1)$ for targets $0,1,2$, respectively.
The matched rows are $(1,1),(1/2,1/2),(1/2,1/2)$;
targets $3,4$ have zero rows.
With fewer than three carriers, no reinforcement activates.
With four carriers, the sole eligible target has all its sources
infected, so equal row sums imply equal hazards. With five carriers
there is no eligible target.
Conditional on three carriers, the ten configurations are equiprobable.
For target $1$, configurations $023$ and $024$ change the hazard pair
from $(\Lambda_1+g,\Lambda_2)$ to $(\Lambda_1+g/2,\Lambda_2+g/2)$.
Target $2$ has the same change for configurations $014$ and $013$.
All remaining eligible target--configuration pairs are unchanged.
For either changed pair the sum of acquisition gains is
\begin{align}
F&:=e^{-\Lambda_1-g}+e^{-\Lambda_2}-e^{-\Lambda_1-g/2}-e^{-\Lambda_2-g/2}\\
&= (1-e^{-g/2})(e^{-\Lambda_2}-e^{-\Lambda_1-g/2}).
\end{align}

Targets $1$ and $2$ each contribute a total acquisition gain $F$ across the ten equally likely three-carrier configurations.

Averaging over configurations and five vertices, then canceling the common recovery contributions, we obtain $2F/(10\cdot5)=F/25$. Multiplying by $q_3$ gives the first statement. For independent Bernoulli($\pi$) initialization, $q_3=\binom{5}{3}\pi^3(1-\pi)^2=10\pi^3(1-\pi)^2$.
\end{proof}

This calculation gives a sufficient mechanism for an acquisition gain:
the redistribution balances two total hazards while preserving their sum.
More generally, if nonnegative hazards $U,V$ admit a coupling with
$\mathbb{E}[U\mid V]=V$, the conditional Jensen inequality implies
$\mathbb{E}[1-e^{-V}]\geq\mathbb{E}[1-e^{-U}]$.
This standard concavity argument concerns the full hazard distribution
among eligible targets; equality of coefficient row sums alone is
insufficient. The proposition is a one-step result on the specified
motif, not an ordering theorem for the large-network trajectories or
co-infection prevalence. Exchangeability need not persist under evolution.

At the experimental rates $b=\beta_1=0.005$, $h=\beta_h=0.1$, and $g=\gamma=0.4$, uniform initialization with exactly three carriers corresponds to $q_3=1$ and gives a gain of approximately $0.171$ percentage points in expected next-step single-pathogen prevalence. The proposition supplies an exact mechanism consistent with the acquisition differences surviving random placement in Fig. \ref{fig:allocation-states}: matching mean hazard can leave the nonlinear acquisition probability changed. It does not establish that the full networks satisfy the required hazard-balancing condition, or predict their endpoint co-infection contrasts in Fig. \ref{fig:allocation-alpha}.

Taken together, the comparisons of this section demonstrate sensitivity to the specified allocation of overlap coefficients among source groups in the tested regime.
They do not establish equilibrium transitions, biological validation, or a new general class of directed contagion models.

\section{Discussion}
\label{sec:discussion}

We started from a simple question: which features of group contact must a model keep before it can reliably compare spreading outcomes? Within HyCo-SIS, our results show that the allocation of
higher-order interactions among source groups matters beyond
their aggregate strength. Our main findings can be summarized as follows.
\begin{itemize}
\item {\bf Overlap is a four-person effect.} Overlap reinforcement is exactly a sum of target-dependent four-person interactions, and it admits a symmetric description precisely when no four people share exactly two or three triangles (Proposition \ref{prop:symmetric}). Four-person sets containing exactly two or three triangle faces obstruct an exact symmetric replacement, whereas a complete four-face tetrahedron admits one.
\item {\bf Seed size and overlap affect spreading.} Near the disease-free state only the pairwise channel survives the mean-hazard linearization, yet the outbreak probability rises from about $0.14$ to $0.89$ as the seed fraction grows from $0.005$ to $0.04$. Once groups are active, raising $\gamma$ from $0$ to $0.4$ increases endpoint co-infection by $6.2$ and $14.4$ percentage points at edge reuse $0.25$ and $0.5$, whereas the estimated increase at reuse zero is small and its interval includes zero  (Sections \ref{subsec:seed} and \ref{subsec:overlap-results}).
\item {\bf Matching one number is not enough.} Pairwise projections calibrated to total prevalence spread earlier than HyCo-SIS, overestimate co-infection over the final fifty steps, and turn every held-out run into an outbreak (Section \ref{subsec:projections}).
\item {\bf Allocation matters.} In the tested networks, the target-strength-matched
redistribution raises endpoint co-infection by about
10 percentage points at $\alpha=1$ and 8 at $\alpha=2$.
Proposition~\ref{prop:exchangeable-gain} provides an exact one-step mechanism
consistent with the observed acquisition differences,
based on the concavity of $1-e^{-\lambda}$ (Section \ref{sec:allocation}).
\end{itemize}

The results actually demonstrate something intuitive: two classrooms can have the same amount of shared group contact per student and still spread a second infection very differently, simply because of who shares which group with whom. In the language of the five students of Fig. \ref{fig:mechanism}, it matters whether the two active groups around student $1$ lean on the same classmate. The pairwise projections studied here do not explicitly retain which contacts belong to the same group.

Most experiments use one representative rate setting;
the allocation comparison additionally tests two
triangle-transmission rates and a lower recovery
probability. All simulations use synthetic triangle
hypergraphs and a horizon of 200 steps, so we make no
claims about equilibria, phase transitions or bistability.
The overlap term is a modeling assumption rather than
an established biological mechanism. The disease-free analysis applies to the mean-hazard closure rather than to the finite stochastic process, and the projection results concern the three projections and the calibration target studied here; more flexible pairwise models can do better \cite{tan2026replicating}, and exact projection ambiguity is a different question from approximate prediction \cite{salova2021projection}. Larger hyperedges, longer horizons and more structural regimes are the natural next steps.

 The classroom of Fig. \ref{fig:mechanism} was not chosen by accident: indeed, schools are one of the places where overlapping groups are easiest to picture, and also one of the places where they have actually been measured. The SocioPatterns high school dataset records contacts among students in nine classes of a high school in Marseilles over five days in December 2013, with anonymous identifiers and 20-second resolution \cite{sociopatterns2013data,mastrandrea2015contacts}. These are pairwise proximity measurements, not observations of how an infection spreads, but they are exactly the kind of raw material from which the hypergraphs studied here could be built.

A natural way to do so would be to take the pairwise contact graph of each time bin and to record a triangle only when all three pair contacts occur in that same bin, keeping track of which triangles share a pair. Aggregating the whole week first is tempting, but it would glue together contacts that never happened at the same time, so any results should be checked against the choice of time bin and of the required co-occurrence duration. Running HyCo-SIS on such data would also require a mapping from seconds of contact to transmission and recovery probabilities, since our synthetic parameters are not calibrated physical rates.

Even without infection histories, which these files do not contain, contact-derived hypergraphs would let us ask whether real classrooms look more like our low-reuse or our high-reuse ensembles, and how much the allocation of overlap among real groups changes co-infection. We think this is a promising direction, and one we hope to come back to.

Finally,  HyCo-SIS can be read as a stochastic system with two kinds of activity at each site: infection switches an activity on, recovery switches it off, and a co-infected node carries both. Pairwise contacts and group activation are two distinct routes through which the activity of neighbors switches on new activity, and the overlap term adds a third route that depends on how the groups themselves intersect.

From this point of view, the redistribution experiments ask a question that is familiar in statistical physics: if every site keeps the same total interaction strength, does the collective behavior stay the same? At least over the finite horizons we simulated, it does not, and the arrangement of interactions among source groups changes the outcome. The projection experiments ask the companion question of which collective observables a simpler interaction model can still reproduce. Our analysis concerns the specific HyCo-SIS overlap rule, and should be distinguished from directed contagion models and from symmetrization methods \cite{directedcontagion2024,directedsync2022}; still, we suspect the lesson is more general, namely that the way many-body interactions are arranged can matter as much as how strong they are.

\section*{Declarations}

\smallbreak
\noindent\textbf{ Code and data.}
The experiments use synthetic networks. Our code and data are available at \href{https://github.com/KatrinaLiu6/hypergraph-co-infection}{https://github.com/KatrinaLiu6/hypergraph-co-infection}.
\smallbreak
\noindent\textbf{Acknowledgements.} 
 LPS is grateful to Magdalen College and the Mathematical Institute at the University of
Oxford for their hospitality and support during the Visiting Fellowship in the course of
which this work was completed. The authors are thankful to the MIT PRIMES-USA program for its support and the opportunity to conduct this research together.

\smallbreak
\noindent\textbf{Competing interests.} The authors have no relevant financial or non-financial interests to disclose.
\smallbreak
\noindent\textbf{Author contributions.}
Conceptualization: Katrina Liu, Laura P. Schaposnik; Methodology and formal analysis:  Katrina Liu; Software and numerical investigation:  Katrina Liu; Writing: Katrina Liu, Laura P. Schaposnik; Supervision: Laura P. Schaposnik. Both authors read and approved the final manuscript.
\smallbreak
\noindent\textbf{AI-assisted technologies.} ChatGPT (OpenAI) and Claude (Anthropic) were used during manuscript preparation to assist with language refinement. 
The authors used OpenAI ChatGPT 5.6-Sol, Codex, and Grammarly in September 2026 for language polishing, grammar checking and general consistency checks. Throughout the research,  the authors used ChatGPT and Codex for code editing and debugging averaging approximately one session per week.  All AI-assisted output was reviewed and edited by the authors, taking responsibility for the final manuscript.

\smallbreak
\noindent\textbf{Funding.} The research of LPS is partially supported by NSF FRG Award DMS2152107 and NSF CAREER
Award DMS 1749013, and also supported in part by grants from the NSF (DMS2235451) and
Simons Foundation (MP-TMPS-00005320) to the NSF-Simons National Institute for Theory and
Mathematics in Biology (NITMB), and by a Simons Travel Award.
\begin{figure*}[t!]
\centering
\includegraphics[width=.88\textwidth]{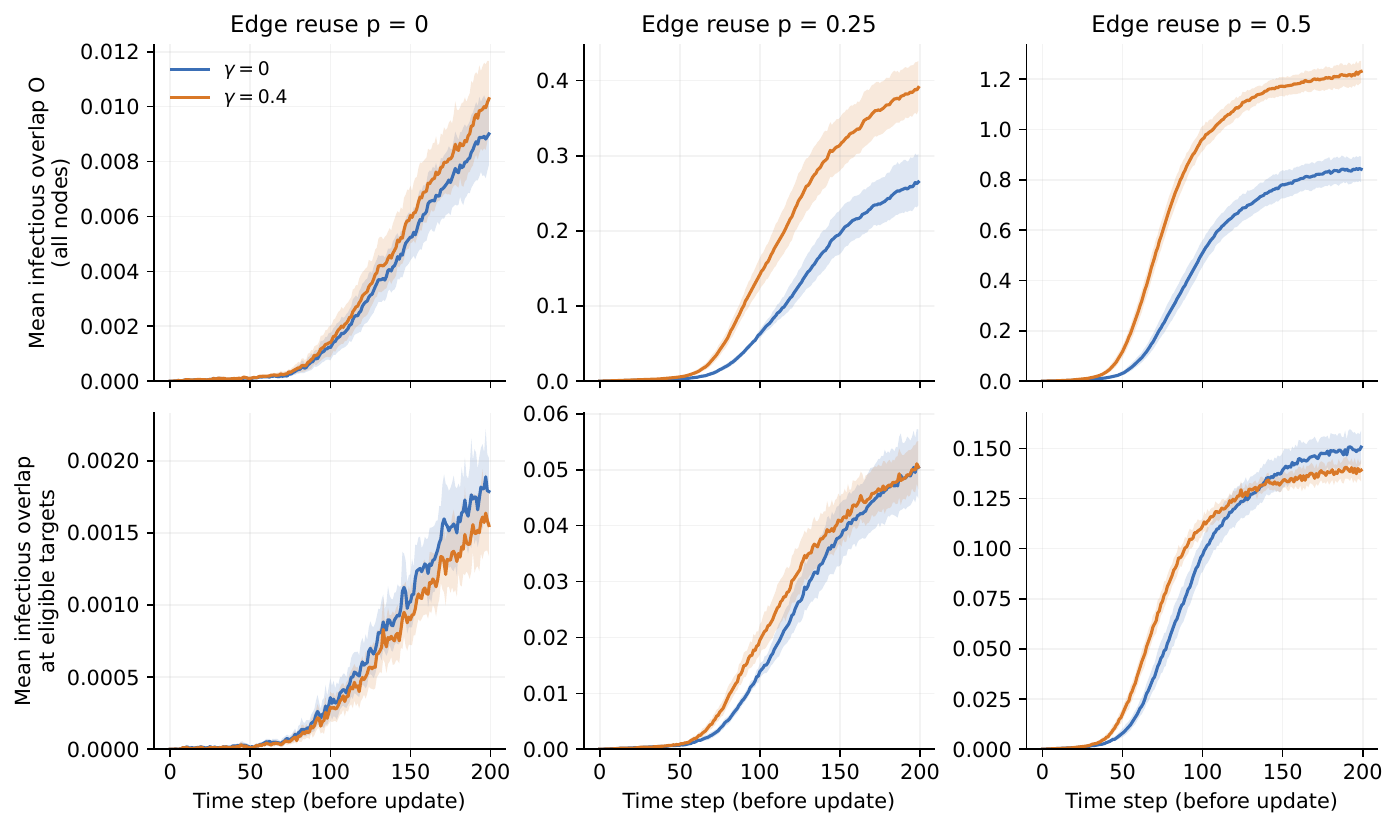}
\caption{Infectious overlap before each update for $\gamma=0$ and $\gamma=0.4$ at edge-reuse probabilities $0$, $0.25$, and $0.5$. Top panels show overlap averaged over all targets; bottom panels show overlap at targets lacking the relevant pathogen, normalized by the full population size. Values are averaged over the two pathogens and trials within each graph. Bands are pointwise $95\%$ bootstrap intervals obtained by resampling whole graph trajectories. The panels use different vertical scales.}
\label{fig:overnight-pressure}
\end{figure*}
\appendix
\section{Infectious overlap trajectories}
\label{app:overlap}

Figure \ref{fig:overnight-pressure} reports the infectious overlap
$O$ recorded before each update in the reinforcement experiments
of Section \ref{subsec:overlap-results}. The all-target averages
are larger on the higher-reuse ensembles. Infectious overlap can
be nonzero at $\gamma=0$, although its contribution $\gamma O$
to the infection pressure then vanishes.

\smallbreak 
\renewcommand{\bibsection}{\section*{References}}
\bibliography{mybib}

\end{document}